\ifdefined\XeTeXversion\else\pdfoutput=1\fi

\documentclass[10pt,twocolumn,letterpaper]{article}

\usepackage[pagenumbers]{cvpr}

\usepackage[dvipsnames]{xcolor}
\definecolor{cvprblue}{rgb}{0.21,0.49,0.74}
\usepackage[pagebackref,breaklinks,colorlinks,citecolor=cvprblue]{hyperref}

\def\paperID{47}
\def\confName{3DV\xspace}
\def\confYear{2027\xspace}

\graphicspath{{figures/}}

\newtheorem{lemma}{Lemma}

\title{Topology-Aware Differentiable Triangle-Soup Reconstruction\\
       via Persistent Homology}

\author{Viritphon Chongpermwattanapol$^{*}$ \qquad Nattapat Damnernyut \qquad
  Pizzanu Kanongchaiyos\\[4pt]
  Department of Computer Engineering, Faculty of Engineering,\\
  Chulalongkorn University, Bangkok, Thailand\\[2pt]
  \texttt{viritphon.1234@gmail.com}\quad\texttt{starter2157@gmail.com}\quad\texttt{pizzanu.k@chula.ac.th}
  \\[6pt]\normalfont\small $^{*}$Corresponding author.\quad Preprint:
  main paper (\S1--\S9) plus the supplementary material as
  Appendices~A--H.}

\begin{document}

\makeatletter
\twocolumn[{%
  \@maketitle
  \centering
  \vspace{-1.9em}
  \captionsetup{type=figure}
  \includegraphics[width=0.74\textwidth]{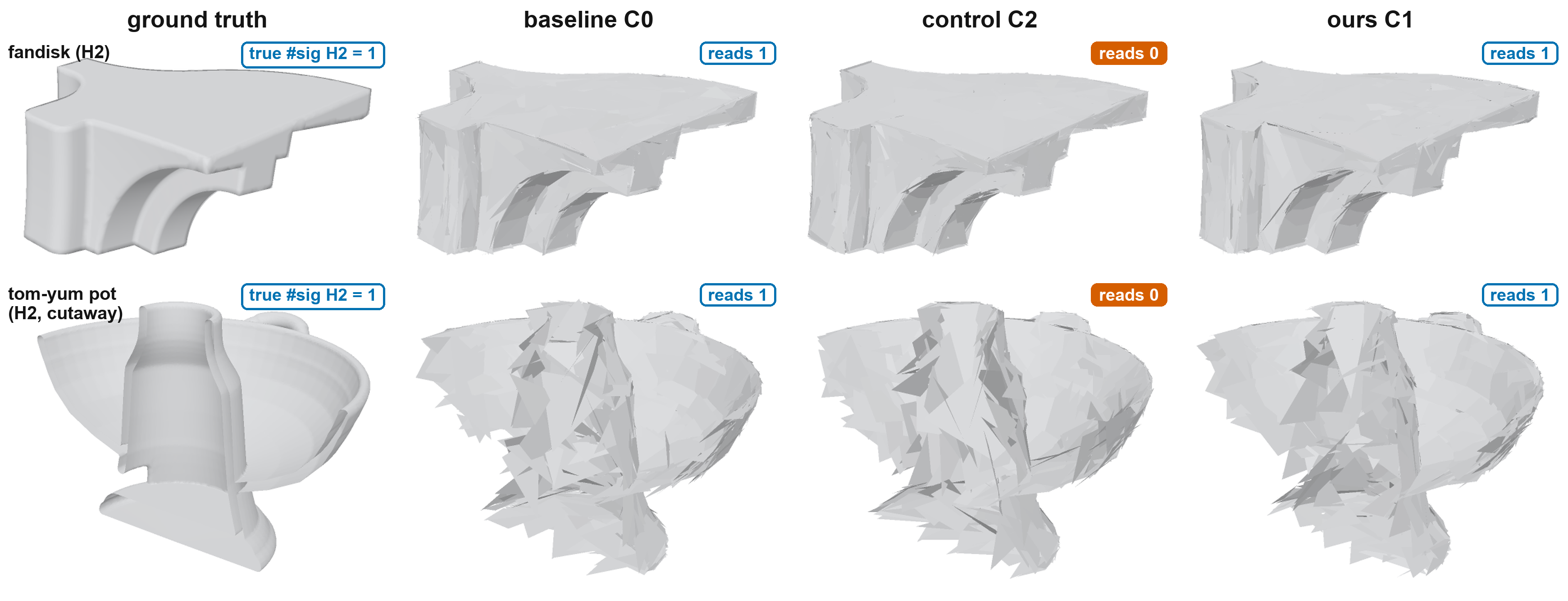}
  \captionof{figure}{Geometry alone cannot distinguish topology;
  measurement can. Final reconstructions (seed 0), identical budgets: ground truth;
  baseline (C0); the norm-matched \emph{non-topological} repulsion
  control (C2); the persistence loss (C1, ours) --- equal-or-better
  Chamfer on every row (\S\ref{sec:generality}). A topological error need not be visible, so
  each cell is badged with its \emph{measured} significant-feature
  count (blue outline = correct, filled orange = wrong; identical
  across seeds): the control erases both voids --- and collapses a loop
  on bob (suppl.\ Fig.~S8) --- while the loss reads correct everywhere,
  cutting tail error up to $10.4\times$ (group means
  $1.5\times$/$4.9\times$; Tables~\ref{tab:main}--\ref{tab:gen}).}%
  \label{fig:matrix}
  \vspace{-0.1em}
}]
\makeatother

\begin{abstract}
Differentiable triangle-soup reconstruction inherits a limitation from its
objective: photometric and geometric losses cannot measure topology, so a
reconstruction with a collapsed loop or a punctured enclosed void can
score exactly as well as a correct one (on Chamfer-equal probes the
diagrams differ $35$--$40\times$ in bottleneck distance).
The standard implicit remedy --- steer \emph{where} the resampler
spends its budget --- does not repair this: in a controlled allocation
study, a topology-informed prior is largely matched by an equally wide
random one, and no prior shape repairs loops.
We therefore move topology into the objective: a differentiable
persistence term compares the evolving surface's diagram, measured on
live surface samples, to a fixed target; gradients flow through a
pair-frozen backward re-expressing matched birth/death simplices as
closed-form circumradii (exact for Gabriel simplices, verified per
pairing), plus a recruitment term restoring the gradient optimal
matching provably lacks when a target feature is missing; one ratio
knob calibrates the loss against the photometric gradient, no
curriculum needed.
Every claim passes a channel-controlled verdict: the loss must beat a
norm-matched \emph{non-topological} control through the identical
gradient channel, at Chamfer parity. Under that rule the loss is
topology-specific for enclosed voids ($4.0$--$7.9\times$ lower error) and
--- the class every allocation prior failed --- for loops ($2.3\times$,
zero phantom handles, while the control collapses one); loss and prior
compose; component counts (H0) are a null result, matched by the
generic control. The verdicts replicate without per-shape tuning on eight external
genus-known meshes in two observable-class groups (pre-registered group
means: loops $1.52\times$, voids $4.87\times$; one a raw artist mesh,
ingested and \emph{certified}; the one non-pass is a no-headroom null,
its already-correct baseline left unharmed), degrading gracefully under
sensor noise.
Scope: all evidence is synthetic and single-machine --- closed surfaces
plus a first open-surface probe --- with the target diagram known in
advance, the regime where correction can be measured exactly; real
scans and blind capture are future work.
Prescription: correct topology in the loss, allocate wide, combine.
\end{abstract}

\section{Introduction}
\label{sec:intro}

Differentiable rasterization reconstructs a triangle soup from images by
following photometric gradients~\citep{tojo2026diffsoup}, periodically
resampling under a fixed budget. Photometric losses are
topologically blind: reconstructions differing in genus, connectivity,
or enclosed voids can be indistinguishable to Chamfer and Hausdorff
while persistence diagrams separate them $35$--$40\times$ in bottleneck
distance on Chamfer-equal probes (\S\ref{sec:related}; suppl.\
Table~S1). Fig.~\ref{fig:matrix} previews
the consequence: identical budgets, equal-or-better Chamfer, topology
standing or falling with the objective.

A decade of connectivity-free reconstruction --- point splats,
Gaussians, and lately a return to triangle soups for their native
graphics-pipeline
fit~\citep{held2025trianglesplatting,burgdorfer2025radiant,tojo2026diffsoup}
--- has treated topology errors as symptoms of geometric
under-provision: represent the surface better, or allocate the budget
more wisely, and topology should follow. We tested the allocation half first (suppl.\ \S A).
It fails: a topological prior helps enclosed voids, but any equally-wide
\emph{non-topological} prior reproduces most of the gain, and no prior
shape repairs loops --- concentrating manufactures phantom handles,
spreading never beats a random control.

The null result relocates the root cause: not \emph{where} triangles are
spent, but \emph{what the objective measures} --- a photometric loss is
a local, pixel-space signal; a loop or an enclosed void is a global,
categorical property of the surface. This paper therefore moves topology
into the \emph{objective}: a differentiable persistence term whose
gradients act directly on features' birth/death values rather than on
where triangles respawn.

Three research hypotheses structure the study, each tested by a
pre-designed condition:
\begin{enumerate}
\setlength{\itemsep}{1pt}
\item \textbf{RH1 (diagnosis)}: topology errors originate in the
  training objective, not in primitive allocation --- tested by the
  allocation study above.
\item \textbf{RH2 (specificity)}: a persistence term improves topology
  in ways no norm-matched non-topological regularizer through the
  identical channel can reproduce --- the control C2
  (\S\ref{sec:conditions}).
\item \textbf{RH3 (complementarity)}: objective and allocation are
  complementary channels, not substitutes --- the composition C5.
\end{enumerate}
Concretely, we contribute \textbf{a matched persistence loss with a
pair-frozen backward} (\S\ref{sec:method}): exact persistence forward
(GUDHI alpha complex), backward through closed-form circumradii ---
Gabriel-exact for 100\% of significant pairs in practice;
\textbf{a recruitment term for missing features}: optimal matching
yields zero gradient when a target feature has no live counterpart;
recruitment attaches the nearest live bar, and removing it collapses
the loop-class win to baseline (\S\ref{sec:results}); \textbf{one calibrated knob}: $\lambda$ set by a
gradient-norm ratio $\rho$, flat over a $10\times$ range, no curriculum
needed --- an ablation, not an assumption;
and \textbf{channel-controlled evaluation}: every claim must beat a
norm-matched \emph{non-topological} control through the identical
channel (\S\ref{sec:conditions}).

The loss proves topology-specific on the class no prior ever won ---
loops, zero phantom handles --- confirming RH2; the channels compose
(RH3); the H0 class is matched by the generic control, sharpening RH1:
topological guidance earns its keep where geometry and topology
decouple.

\section{Related work}
\label{sec:related}

\paragraph{A decade of gradient-based reconstruction: four positions on
connectivity.}
OpenDR~\citep{loper2014opendr}, neural mesh
rendering~\citep{kato2018neural}, SoftRas~\citep{liu2019softras},
DIB-R~\citep{chen2019dibr} and nvdiffrast~\citep{laine2020nvdiffrast}
established gradient flow from pixels to geometry; reconstruction work
since splits on what that geometry should be.
(i)~\emph{Implicit fields} --- learned occupancy and signed
distance~\citep{mescheder2019occupancy,park2019deepsdf}, volume-rendered
from NeRF~\citep{mildenhall2020nerf} through NeuS~\citep{wang2021neus} to
Neuralangelo~\citep{li2023neuralangelo} --- are watertight and
topologically consistent \emph{by construction} (the strongest
counter-position), at field-evaluation and extraction cost.
(ii)~\emph{Fixed-connectivity templates} deformed per
input~\citep{wang2018pixel2mesh,hanocka2020point2mesh} keep manifoldness
but freeze topology at the template's genus.
(iii)~\emph{Differentiable iso-surfacing} re-extracts the mesh from a
volumetric scaffold at every
step~\citep{shen2021dmtet,munkberg2022nvdiffrec,shen2023flexicubes}
(Shape-as-Points bridges point soups to watertight
implicits~\citep{peng2021sap}) --- topological freedom at grid cost.
(iv)~\emph{Connectivity-free primitive soups} --- point
splatting~\citep{yifan2019dss}, patch atlases~\citep{groueix2018atlasnet},
Gaussians~\citep{kerbl2023gaussians} and their surface-accurate
variants~\citep{huang20242dgs,guedon2024sugar}, and a 2025--26 return to
native triangles (splatted~\citep{held2025trianglesplatting},
soft-connected~\citep{burgdorfer2025radiant}, rasterized by
DiffSoup~\citep{tojo2026diffsoup}) --- are the fastest and the easiest
to budget, but even their surface-minded members regularize
\emph{geometry} (alignment, flatness, connectivity forces), never a
measured topology.
MeshSplatting, conversely, renders \emph{already-connected} opaque
meshes differentiably~\citep{held2026meshsplatting}: connectivity as
input; ours a free soup with measured topology.
(i) and (iii) solve topology representationally; (iv) had not addressed
it. Cut by \emph{optimization philosophy} rather than representation,
the decade holds \emph{geometry-oriented} objectives (every method
above), \emph{sampling-oriented} steering of where a fixed objective
spends primitives (below), and an empty third family ---
\emph{topology-oriented} objectives that measure the property and
differentiate through it (suppl.\ Table~S6). Prior work occupies every
pairwise intersection of the section's axes --- explicit triangles,
differentiable optimization, topology in the objective; this paper
takes the three-way one, at the soup's budget and speed.
We build on DiffSoup, whose in-loop
prune/respawn resampler is the very mechanism our allocation study
(suppl.\ \S A) steered; here the resampler stays untouched and only the
objective changes. Vertex-domain preconditioning~\citep{nicolet2021largesteps} is
complementary: it changes \emph{how} vertices move, we change \emph{why}.

\paragraph{The allocation channel: why priors alone fail.}
Three findings of our allocation-channel study set this paper's problem.
(1)~Geometric metrics are topology-blind: on bisection-matched probes with
Chamfer equal \emph{by construction} ($0.916$/$0.628$/$0.977\%$ per
pair), the discriminating-dimension bottleneck separates correct from
broken topology by $35$--$40\times$ --- thin handle ${\approx}0$ vs
$.0302$, bridged merge $.0014$ vs $.0574$, punctured void $.0040$ vs
$.1385$ --- and Hausdorff$_{95}$ even \emph{prefers} the wrong candidate
in two of three cases (suppl.\ Table~S1).
(2)~For voids, a spread topological prior
cuts the bottleneck $33$--$53\%$ below baseline --- but a width-matched
\emph{random} control recovers most of that gain, leaving a small,
shape-dependent residual: the prior's value is carried largely by its
width. (3)~No prior shape repairs loops: the best torus arm is the
\emph{non-topological} wide control ($.0267$ vs the spread prior's
$.0316$), and the concentrated field manufactures phantom handles
($4.4$ significant loops vs the true $2$); C5 below reuses the study's
strongest field (B4; suppl.\ \S A).

\paragraph{Persistence in learning objectives.}
Persistent homology~\citep{edelsbrunner2002persistence} with its stability
guarantees~\citep{cohensteiner2007stability} has been used as a trainable
signal via differentiable sorting of filtration
values~\citep{gabrielsson2020topologylayer}, for topology-preserving image
segmentation~\citep{hu2019topopreserving,clough2022topoloss}, point-cloud
and shape optimization and 3D volume
prediction~\citep{bruel2020toporecon,poulenard2018shapematching,waibel2022topological},
and diagram-metric optimization with convergence
analysis~\citep{carriere2021optimizing}; barcode-valued maps have since been
given a general differential framework~\citep{leygonie2022framework}, and the
per-step cost of persistence-guided descent attacked
directly~\citep{nigmetov2024bigsteps} --- yet never as the objective of
an \emph{explicit triangle-based} reconstruction.
Nearest our setting, concurrent work brings persistence to
reconstruction at other, complementary entry points: an image-space barcode term for
Gaussian splatting~\citep{shen2025topologygs}, homology-guided camera
placement~\citep{gao2026homology} or a genus-matched
template~\citep{gao2025genus} in mesh inverse rendering, and a
connected-component constraint on implicits fitted to point
clouds~\citep{jignasu2024stitch}.
Partition our setting's three requirements --- (i)~a topological quantity
of the evolving \emph{surface} measured differentiably in the objective,
(ii)~image-based training through a renderer, (iii)~a fixed primitive
budget --- and each neighbour satisfies at most one (suppl.\ \S D.7): the
image-space barcode term never measures the surface (i); the point-cloud
pipeline forgoes renderer and budget (ii, iii); the template line
\emph{assumes} the genus we measure (i). A head-to-head benchmark is
therefore ill-posed (inputs, outputs, and cost regimes differ); our
evaluation instead pits the loss against norm-matched controls through
its own channel (\S\ref{sec:setup}).
The nearest soup-side neighbour regularizes coherence with \emph{soft
connectivity forces}~\citep{burgdorfer2025radiant} --- local continuity
pressure, no measured topology; we replace that heuristic with the
measurement. Our backward pass is in the persistence family
--- unlike work that back-propagates through filtration-value sorting or
barcode coordinates in general position, we re-express each paired
simplex's birth/death as a closed-form circumradius --- with three
pipeline-specific ingredients.
(i)~\emph{Alpha-complex pairing over $\alpha{\times}$area surface samples}
of a live triangle soup --- not image grids or raw clouds --- couples
the diagram to the geometry actually being trained.
(ii)~A \emph{Gabriel-exactness gate}: the circumradius re-expression
equals the filtration value only for Gabriel simplices, so we verify it
per pairing (relative $10^{-6}$) --- certifying the frozen gradient
rather than assuming it.
(iii)~\emph{Recruitment for missing features}: optimal partial matching
has zero gradient exactly when a target feature is entirely absent, so
we trade metric fidelity for a restored descent direction
(\S\ref{sec:method}), quantifying both its pathology and its necessity
(\S\ref{sec:results}).

\paragraph{Topology-aware geometry processing.}
Topology-aware reconstruction and repair predate differentiable
pipelines~\citep{sharf2007topologyaware,ju2004repair,attene2010repair}, and
budgeted remeshing has a long
lineage~\citep{hoppe1996progressive,garland1997qem,alliez2003anisotropic,botsch2004remeshing,dunyach2013adaptive};
our contribution is orthogonal --- the budget mechanics stay DiffSoup's,
only the training signal becomes topology-aware.

\section{Method: a matched persistence loss for \mbox{triangle} soups}
\label{sec:method}

\subsection{Objective}
DiffSoup trains with a photometric objective $L_{\mathrm{photo}}$ (opacity
auxiliary + $0.8\,L_1 + 0.2\,\mathrm{SSIM}$)~\citep{tojo2026diffsoup}, which
we adopt \emph{unchanged}: every ablation below isolates the topology
channel, never photometric tuning. We add
\begin{equation}
L \;=\; L_{\mathrm{photo}} \;+\; \lambda(t)\, L_{\mathrm{topo}}(V),
\end{equation}
where $V$ are the soup's vertex positions and $L_{\mathrm{topo}}$ compares the
persistence diagram of the live surface to a fixed \emph{target bundle}
$\mathcal{T}$: the significant finite bars of the ground-truth shape's
diagram, computed once at the density $M$ the loss will see
(alpha values shift with spacing: density matching is a correctness
requirement), diagonal-normalized, and locked for all live evaluations.
Throughout, a bar is \emph{significant} iff its lifetime exceeds the
measurement floor $6\,r_{\mathrm{med}}(M)$ ($r_{\mathrm{med}}$ = median
nearest-neighbour sample spacing; derived and validated in
\S\ref{sec:discussion}).

\subsection{Live diagram and pair-frozen backward}
At each refresh we draw $M$ surface samples ($M{=}2048$ unless the
measurement-floor rule of \S\ref{sec:setup} dictates a denser bundle)
$X_j=\sum_i w_{ij}\,V[F(f_j),i]$ with faces $f_j$ drawn
$\sigma(\alpha)\times$area-weighted and barycentric weights $w_{ij}$ frozen;
$X$ is recomputed from live $V$ every step, so gradients reach vertices while
the combinatorics stay fixed. The alpha complex of
$X$~\citep{edelsbrunner1994alphashapes}, computed exactly with
GUDHI~\citep{maria2014gudhi}, yields persistence
pairs (birth simplex, death simplex) per dimension (suppl.\ Fig.~S9). For a Gabriel simplex the
alpha filtration value is its squared circumradius, so each paired bar's
coordinates are closed-form functions of a few sample positions, all
exactly differentiable:
\begin{equation}
\begin{gathered}
R_{\mathrm{edge}}=\tfrac12\lVert p_1{-}p_0\rVert,\qquad
R_{\mathrm{tri}}=\frac{\lVert u\rVert\,\lVert v\rVert\,\lVert u{-}v\rVert}
                      {2\,\lVert u\times v\rVert},\\[2pt]
R_{\mathrm{tet}}=\lVert y\rVert\quad\text{with}\quad
A\,y=\tfrac12\operatorname{diag}\!\big(AA^{\!\top}\big),
\end{gathered}
\end{equation}
where $p_0,\dots,p_k$ are the $k$-simplex's sample positions,
$u=p_1{-}p_0$, $v=p_2{-}p_0$, and $A$ stacks the rows $p_i{-}p_0$
(the circumcenter is $p_0{+}y$). Edges carry H0 deaths and H1 births,
triangles H1 deaths and H2 births, tetrahedra H2 deaths.
We verify, per pairing, that the re-expressed value matches the filtration
value (relative $10^{-6}$); a pair failing this Gabriel test is
\emph{detached} --- it keeps its (correct) value in the loss but
propagates no gradient, so an exactness failure can never inject a
wrong one.
In practice the gate never fires: 100\% of significant pairs are
Gabriel-exact on all six shapes' clean clouds, and the per-refresh
failure count is zero at every one of the $5{,}500{+}$ logged refreshes
(all $2{,}400$ noisy ones included, \S\ref{sec:results}) --- a constant
pass rate, not an aggregate. Stronger fallbacks (dropping pairs, subgradients, locally
raising $M$) remain available.
The pairing is treated as locally constant --- the standard pair-frozen
device~\citep{gabrielsson2020topologylayer,carriere2021optimizing} ---
refreshed every $K{=}10$ steps and, necessarily, after every resampling
event (which replaces the primitive tensors wholesale); $K{=}10$
is an amortization choice, not a tuned knob --- it bounds refresh cost
($174$\,ms each; the fixed ${\approx}46$\,s of \S\ref{sec:results})
while the pairing stays local.

\subsection{Matching, diagonal pressure, and recruitment}
Let $D$ be the live significant bars and $\mathcal{T}$ the target bars of one
dimension. An optimal partial matching (Hungarian on the standard augmented
cost; equivalently $W_2$) splits bars into matched pairs, unmatched live, and
unmatched target:
\begin{equation}\label{eq:matching}
\begin{split}
L_{\mathrm{topo}} = \!\!\sum_{(i,j)\,\mathrm{matched}}\!\!\!
  \big[(b_i-b^*_j)^2+(d_i-d^*_j)^2\big]\\[-2pt]
 \;+\; \sum_{i\,\mathrm{unmatched}} \Big(\tfrac{d_i-b_i}{2}\Big)^{\!2}
 \;+\; L_{\mathrm{recruit}} .
\end{split}
\end{equation}
The first term pulls matched features to their targets; the second
pushes spurious features to the diagonal; the third exists because
optimal matching has a provable failure mode:
\begin{lemma}\label{lem:zerograd}
If a target bar $t\in\mathcal{T}$ is unmatched under the optimal partial
matching, then $L_{\mathrm{topo}}$ without $L_{\mathrm{recruit}}$ is
locally independent of $t$: no component of its gradient in $V$ acts to
create the missing feature.
\end{lemma}
\noindent\emph{Proof sketch.} The unmatched target's augmented cost is its
diagonal penalty $\big((d^*_t-b^*_t)/2\big)^2$, constant in the live
coordinates; every $V$-dependent term of Eq.~\ref{eq:matching} references
live bars only, and the matching is locally constant in $V$ (the
pair-frozen regime) --- no gradient term is directed at $t$.\hfill$\square$

Recruitment restores a descent direction. Let $\mathcal{T}_{\mathrm{miss}}$ be the unmatched target
bars, processed in decreasing persistence order, and let the candidate pool
$\mathcal{A}$ hold every live bar not claimed by the matching --- unmatched
significant bars and, crucially, \emph{sub-threshold} bars, which matching
never sees. Each missing target greedily claims its nearest pool bar,
without replacement, under the same squared birth--death metric as the
matched term:
\begin{equation}
\begin{split}
L_{\mathrm{recruit}} &= \!\!\sum_{j\in\mathcal{T}_{\mathrm{miss}}}\!\!
  \big[(b_{k(j)}-b^*_j)^2+(d_{k(j)}-d^*_j)^2\big],\\
k(j) &= \operatorname*{arg\,min}_{k\,\in\,\mathcal{A}_j}
  \big[(b_k-b^*_j)^2+(d_k-d^*_j)^2\big],
\end{split}
\end{equation}
where $\mathcal{A}_j$ is the pool remaining after earlier recruitments;
recruited bars are excluded from the diagonal term and back-propagate
through the same circumradius machinery.
This deliberately departs from metric-faithful $W_2$;
\S\ref{sec:results} probes its location-blindness pathology directly,
and C6 shows the term is load-bearing.
Greedy rather than optimal assignment is deliberate: in every logged
regime the missing-target set is a single bar (torus: one unreached
target at every refresh, every seed; \S\ref{sec:results}), where the two
coincide; when they diverge, decreasing-persistence order spends the
pool on the most significant features first.

\subsection{Calibration instead of curriculum}
$\lambda(t)$ ramps linearly from 20\% to 50\% of training, with peak set once
by a gradient-norm ratio at the first active step:
$\lambda_{\mathrm{peak}}=\rho\,\lVert\nabla_V L_{\mathrm{photo}}\rVert /
\lVert\nabla_V L_{\mathrm{topo}}\rVert$. The single knob $\rho$ is
dimensionless ($\rho{=}0.1$ throughout); the tail error is flat across
$\rho\in\{0.03,0.1,0.3\}$, a constant-$\lambda$ ablation matches the
ramped version, and alternative ramp windows land within seed noise
(suppl.\ Table~S8) --- calibration, not scheduling, makes the knob safe.

\subsection{Conditions}
\label{sec:conditions}
C0: baseline (photometric only). C1: $+L_{\mathrm{topo}}$. C2: $+$ a
short-range repulsion on the \emph{same} frozen samples, calibrated with the
\emph{same} $\rho$ --- the norm-matched control, the RH2 test (C2g: a
gentler radius). C3: C1 without the ramp.
(C4, curvature-weighted, is deferred; the numbering gap keeps
condition tags comparable across studies.)
C5: C1 $+$ the best prior of our allocation study (the spread topological
field B4; suppl.\ \S A) --- the RH3 test.
C6: C1 without recruitment (pure matching $+$ diagonal) --- does the
term matter in-loop, or only for cold-start repair?

\section{Experimental setup}
\label{sec:setup}

\paragraph{Scenes and budgets.} Synthetic COLMAP scenes with analytically
known topology, rendered from 24--72 Fibonacci-sphere viewpoints at
200--256\,px (every-8th test holdout $\Rightarrow$ 21--63 training
views): sphere and cube
(void class H2, $\beta{=}(1,0,1)$, budget $N{=}1200$), torus (loop class H1,
$(1,2,1)$, $N{=}700$), two spheres (component class H0, $(2,0,2)$, $N{=}700$),
and genus-2 double torus as a stretch case ($(1,4,1)$, $N{=}2000$). Budgets
are the allocation study's headroom points (suppl.\ \S A) --- tight
enough that baseline topology is imperfect.
Eight external, genus-known meshes probe generality beyond the analytic
family, organized as the study's two observable classes
(\S\ref{sec:generality}): a \emph{loop group} --- \emph{bob} (genus 1),
\emph{rocker-arm} (1), the CGAL \emph{eight} (2) --- and a \emph{void
group} --- \emph{spot}, \emph{fandisk}, \emph{armadillo}, \emph{horse}
(all genus 0), the \emph{tom-yum pot} (3, below); sources, licenses and
certificates: suppl.\ \S G. Seven pass the mesh certificate below
as-is; five further
candidates failed it (open boundaries, non-manifold junctions) and were
excluded, never repaired. Membership, observable class and density were fixed by the
pre-registered staircase rule \emph{before} any training (rocker-arm:
torus precedent, H1-restricted at $M{=}2048$; eight: the pot's, full
$(1,4,1)$ bundle at $M{=}8192$). Scenes: the same pipeline (48 views,
200\,px), unchanged class budgets (trio pilot: tails $.043$--$.054$,
in the analytic band).
The fourth, the \emph{tom-yum pot} (Fig.~\ref{fig:matrix}, bottom row;
aluminium showcase: suppl.\ Fig.~S7), probes the step
every real deployment needs: a raw artist-authored mesh
rather than a clean benchmark. As exported from Blender it is textbook
triangle soup (9 overlapping open shells, 232 non-manifold junction
edges) with no trustworthy ground-truth topology --- the same defect
that excluded ShapeNet. An ingest-and-certify pass (exact weld,
offset-solidification into a thin closed shell, then \emph{certification
by measurement}: exact simplicial homology cross-checked against
per-component Euler characteristics, both independent of the loss's
sampled alpha complex; watertightness; bit-identical rebuilds) yields
one body of genus 3, $\beta{=}(1,6,1)$ exact,
rendered and budgeted as a void-class shape ($N{=}1200$). This ingest
pipeline is a test fixture here, not a claimed contribution. Where the
other externals arrive clean enough to certify as-is, the pot's
topology exists only \emph{after} ingest --- a controlled real-world
proxy, not a blind test: what real scans will need (mechanics and
challenge inventory: suppl.\ \S B, Table~S4).
The pot also exercises the floor rule of \S\ref{sec:discussion}
prospectively: every H1 loop of the certified shell sits below the
$6\,r_{\mathrm{med}}$ floor at every working $M$, while the narrow flue
mouth caps at small $\alpha$ --- an H2 void on the pot axis, first
clearing the floor at $M{=}8192$ (margin $1.20$--$1.23\times$, five
sampling seeds). The pre-registered rule builds the bundle there (a
floor-rule density deviation later repeated by eight; the loss samples
at the same density): exactly one significant bar, stable to
re-sampling (${\sim}10^{-5}$).

\paragraph{Protocol.}
All conditions share one implementation: C0 is the unmodified DiffSoup
trainer; the loss enters as one additive term, and disabled it executes
none of the added code (no renderer or resampler edits).
2{,}500 steps; resampling every 100 steps (untouched);
$\rho{=}0.1$, ramp $0.2{:}0.5$, refresh $K{=}10$, $M{=}2048$; decisive shapes
(sphere, torus) and the double torus run 5 seeds (the former with the
full condition set incl.\ C6), replication shapes (cube, two spheres)
3 seeds (C0/C1/C2).
(C7/C7h $=$ C1 $+$ noise; decisive shapes, 3 seeds) perturbs the cloud
the loss \emph{plan} sees with i.i.d.\ Gaussian noise at every refresh:
the plan decides on the noisy observation, the pulls apply to the true
live positions.
The torus restricts the loss to H1 (its own void bar is sub-threshold at
$M{=}2048$ --- a loss must not act on features it cannot reliably
measure); the double torus likewise targets its two measurable loops of
four. The generality wave runs C0/C1/C2, three seeds, each shape at its
staircase-fixed bundle (bob: full loops$+$void; rocker-arm: H1 only;
eight and the pot: $M{=}8192$).

\paragraph{Metrics and verdict.} Primary: bottleneck to the target
diagram in the discriminating dimension, tail-averaged (last three
dumps, steps $2{,}300$--$2{,}500$), seed-averaged (20k eval samples,
target-normalized).
Secondary: significant-feature counts (phantom check) and Chamfer parity.
Reporting policy: mean$\pm$sd over seeds; Welch $\sigma$ only where both
arms ran five seeds; three-seed verdicts rest on disjoint per-seed
ranges (suppl.\ Table~S7).
\emph{Verdict rule}: C1 must beat C0 \textbf{and} every norm-matched
control at Chamfer parity --- a topological win a dumb regularizer can
match is not one. \emph{Chamfer parity}, formally: the passing arm's
seed-mean Chamfer must satisfy
$\mathrm{Ch}(\mathrm{C1})\le 1.15\,\mathrm{Ch}(\mathrm{C0})$; the $15\%$
headroom absorbs seed variance and nondeterminism
(\S\ref{sec:limitations}) and never did any work --- every
passing arm's Chamfer was at or below baseline.

\section{Results}
\label{sec:results}

A gate first: the loss alone (no photometric term, Adam on points)
repairs all four controlled probes --- a punctured void closes
($8.8\times$), a spurious handle dies at the diagonal, mis-spaced
components separate ($1229\times$), a bridged merge with \emph{no}
significant live H0 bar (Lemma~\ref{lem:zerograd}'s case) is severed
($450\times$, $\beta_0\,1{\to}2$); the predicted location-blindness did
not materialize (suppl.\ \S C, Fig.~S5). Raw-cloud sanity toys;
multipliers not comparable to training.

\subsection{Controlled comparison: analytic shapes}
\label{sec:cmatrix}
Table~\ref{tab:main} (the \emph{C-matrix}: every condition against every
analytic shape) reports tail bottleneck (mean$\pm$sd over seeds) in each
shape's discriminating dimension; Chamfer in the text where it matters;
the full training trajectories for the two decisive shapes are in
suppl.\ Fig.~S6 (line = seed mean, band = seed min--max).

\begin{table*}[t]
\centering\small
\caption{Tail bottleneck to target (mean$\pm$sd; lower is better);
$\times$ = vs C0; verdict: C1 $<$ C0 and $<$ every control at Chamfer
parity. Seeds: 5 (sphere, torus, double torus --- saturated
identically), 3 (cube, two spheres). C6 = no recruitment.}
\label{tab:main}
\footnotesize
\setlength{\tabcolsep}{4pt}
\begin{tabular}{lllllll}
\toprule
shape (dim) & C0 & C1 loss & C2 control & C3 no-ramp & C5 loss+prior & C6 no-recruit \\
\midrule
sphere (H2)      & .0623$\pm$.0063 & \textbf{.0156$\pm$.0013} (4.0$\times$) & .1372 (worse) & .0151$\pm$.0004 & \textbf{.0082$\pm$.0007} (7.6$\times$) & .0166$\pm$.0005 \\
cube (H2)        & .0582$\pm$.0046 & \textbf{.0074$\pm$.0011} (7.9$\times$) & .1346 ($\beta_2\,1{\to}0$) & --- & --- & --- \\
torus (H1)       & .0424$\pm$.0031 & \textbf{.0180$\pm$.0016} (2.3$\times$) & .0457 ($\beta_1\,2{\to}1$) & .0155$\pm$.0015 & \textbf{.0133$\pm$.0004} (3.2$\times$) & .0411$\pm$.0038 \\
two spheres (H0)& .0065$\pm$.0008 & .0007$\pm$.0002 (9.5$\times$)          & .0008 (ties C1) & --- & --- & --- \\
double torus (H1)& .0264$\pm$.0000 & .0264$\pm$.0000 (saturated)           & .0264$\pm$.0000 & --- & --- & --- \\
\bottomrule
\end{tabular}
\end{table*}

\paragraph{Voids (H2): topology-specific, large.} Sphere $4.0\times$
($16.2\sigma$ Welch vs C0, five seeds), cube $7.9\times$ (three seeds:
ranges disjoint, C1 $.0064$--$.0086$ vs C0 $.0551$--$.0636$), both at
equal-or-better Chamfer (sphere $0.46$ vs $0.50$; cube $0.92$ vs $1.05$).
The
norm-matched control is \emph{destructive}, not merely weaker: it pins
the void's death at its repulsion equilibrium ($2.2\times$ worse,
seed-identical) and erases the cube's void outright; no tested strength
reproduces any gain --- C1 $<$ C0 $<$ C2g $<$ C2
(\S\ref{sec:limitations}; spread: suppl.\ Fig.~S3; mechanism:
\S\ref{sec:discussion}).

\paragraph{Loops (H1): the channel's distinguishing result.} No prior
shape won this class (suppl.\ \S A). The loss cuts torus
H1 error $2.3\times$ with \emph{zero phantom handles in all five seeds}
(count trajectories flat at the true $\beta_1{=}2$,
suppl.\ Fig.~S4), while the control collapses one loop by step
${\sim}800$ and never recovers it. Metric pressure
on birth/death values suits 1-D features in a way budget allocation does
not (mechanism: \S\ref{sec:discussion}).

\paragraph{Components (H0): not topological --- third channel, same answer.}
C1 reduces the (already small) baseline error $9.5\times$, but the
generic control matches it (.0008 vs .0007); the increment appears only
in Chamfer ($0.54$ vs $1.94$). Three channels agree: $\beta_0$ is
carried by shell count and separation.

\paragraph{Composition (C5).} Loss $+$ spread prior is best everywhere
tested --- sphere $.0082$ ($7.6\times$ vs C0, ${\approx}2\times$ vs C1 alone,
${\approx}4.4\times$ below the prior alone%
), torus $.0133$ --- with the best sphere Chamfer (torus $.502$,
trailing only C3's $.495$). Allocation and value-tuning are
complementary, not substitutes.

\paragraph{Stretch case (genus 2): saturated, honestly.} Every double-torus
arm returns exactly $.0264$: the two tube loops are unexpressable at
feasible budgets --- the \emph{representation's} ceiling, not the
channel's (suppl.\ \S B; floors: \S\ref{sec:discussion}).

\paragraph{Ablations.} The $\rho$-response is flat over a $10\times$
range (suppl.\ Table~S8); C3 (no curriculum)
$\approx$ C1 on both decisive shapes \emph{and} off the analytic family
(bob: C3 $.0171{\pm}.0026$ vs C1 $.0208{\pm}.0008$).
Overhead, honestly: the loss adds a \emph{fixed}
${\approx}46$\,s per 2{,}500-step run (refresh $174$\,ms at $M{=}2048$
every 10 steps) --- more than the $33$--$41$\,s photometric baseline on
these deliberately tiny scenes, but the cost scales $O(M\log M)$ per
refresh, not with scene or image complexity, so its share of a run whose
photometric cost is $T$ is $46/(T{+}46)$ --- an arithmetic projection, ${\approx}7\%$ of a ten-minute run ---
independent of image count and resolution; no parameters are added.

\paragraph{Recruitment is load-bearing for loops (C6).} Removing the
recruitment term (pure matching $+$ diagonal) collapses the torus win
entirely: $.0411\pm.0038$, statistically baseline ($0.6\sigma$ from C0,
$12.4\sigma$ worse than C1) at baseline Chamfer, though both loops
still \emph{exist} ($\#$sig H1 $= 2$ every seed; suppl.\ Fig.~S6). The
logs locate the mechanism: one torus loop is significant at every
refresh, the other chronically \emph{below} threshold (all five seeds,
all 200 refreshes) --- recruitment is the second loop's sole gradient
path; on this class it \emph{is} the mechanism, not a cold-start
device. On the sphere it never fires (0 of 200, all seeds): C6 is
loss-identical to C1, and the $1.5\sigma$ tail gap ($.0166$ vs $.0156$)
measures pure run-to-run noise (\S\ref{sec:limitations}).

\paragraph{Sensor-noise stress (C7/C7h).} With the plan's cloud
perturbed at every refresh ($\sigma = 0.5\%$/$1\%$ of the diagonal
${\approx}0.5$/$1.0$ spacings, \S\ref{sec:setup}), the loss degrades
gradually and never harms: torus $.0210{\pm}.0018$ / $.0259{\pm}.0022$
($2.0\times$/$1.6\times$ below baseline, vs C1's $2.3\times$), sphere
$.0235{\pm}.0009$ / $.0328{\pm}.0013$ ($2.7\times$/$1.9\times$, vs
$4.0\times$), at equal-or-better Chamfer, correct counts in every run
(suppl.\ Fig.~S6), zero Gabriel failures in all $2{,}400$ noisy
refreshes (internals: suppl.\ \S B).

\subsection{Generality: external genus-known meshes}
\label{sec:generality}
Everything above is on one analytic family, so we repeat the verdict
protocol --- C0/C1/C2, three seeds, class budgets, $\rho$ and ramp
unchanged, density fixed by the floor rule --- on the eight external
genus-known meshes of \S\ref{sec:setup}'s two groups, with no
per-shape tuning.
Seven of eight \textbf{pass} the verdict rule at equal-or-\emph{better}
Chamfer (Table~\ref{tab:gen}; original-wave trajectories in suppl.\
Fig.~S2); the eighth, rocker-arm, is a \emph{no-headroom null}, not a
counterexample (obs.~5). The control replicates its failure modes: all
five voids erased ($\beta_2\,1{\to}0$, every seed; the pot at
$1.36\times$ worse Chamfer), a bob loop collapsed ($\beta_1\,2{\to}1$),
eight's missing pair left missing.

\begin{table}[t]
\centering
\caption{Generality wave, two observable-class groups (mean$\pm$sd, 3
seeds; $\times$ = vs C0; \emph{group mean} $\pm$ sd over member shapes,
pre-registered). Seven of eight \textbf{pass} Table~\ref{tab:main}'s
verdict rule; rocker-arm is a no-headroom null (counts intact in every
arm). *seed-exact unmatched-target pins (\S\ref{sec:discussion}).
Per-seed values: Table~S7; provenance, certificates, staircases:
suppl.\ \S G.}
\label{tab:gen}
\scriptsize
\setlength{\tabcolsep}{1.5pt}
\resizebox{\linewidth}{!}{%
\begin{tabular}{llll}
\toprule
shape & C0 & C1 loss & C2 control \\
\midrule
bob          & .0426$\pm$.0012 & \textbf{.0208$\pm$.0008} (2.1$\times$) & .0437 ($\beta_1\,2{\to}1$) \\
rocker-arm   & .0083$\pm$.0004 & .0087$\pm$.0015 (0.96$\times$) & .0113 (counts intact) \\
eight (genus 2) & .0249$\pm$.0000* & \textbf{.0162$\pm$.0017} (1.5$\times$) & .0249 (2 of 4 loops) \\
\emph{loop group mean (H1, 3)} & & \textbf{1.52$\times$}$\,\pm\,$0.55 & \\
\midrule
spot         & .0521$\pm$.0081 & \textbf{.0237$\pm$.0000} (2.2$\times$) & .0652 ($\beta_2\,1{\to}0$) \\
fandisk      & .0538$\pm$.0022 & \textbf{.0052$\pm$.0008} (10.4$\times$) & .0571 ($\beta_2\,1{\to}0$) \\
tom-yum pot  & .0221$\pm$.0003 & \textbf{.0057$\pm$.0002} (3.9$\times$) & .0223 ($\beta_2\,1{\to}0$) \\
armadillo    & .0511$\pm$.0000* & \textbf{.0124$\pm$.0001} (4.1$\times$) & .0511 ($\beta_2\,1{\to}0$) \\
horse        & .0408$\pm$.0000* & \textbf{.0109$\pm$.0022} (3.8$\times$) & .0408 ($\beta_2\,1{\to}0$) \\
\emph{void group mean (H2, 5)} & & \textbf{4.87$\times$}$\,\pm\,$3.17 & \\
\bottomrule
\end{tabular}}
\end{table}

Observations.
(1)~\emph{The H1 claim survives the family change}: bob halves loop
error, zero phantom handles every seed (suppl.\ Fig.~S8), replicating
torus.
(2)~\emph{The genus-2 ceiling was representational, not topological}:
on the fat CGAL eight, baseline and control read only two of the four
certified loops (value pinned at the unmatched bound, seed-exact); the
loss \emph{restores the missing pair} ($\#$sig $4/3/4$) and cuts value
error $1.5\times$ at $0.89\times$ Chamfer --- the double torus's
saturation (\S\ref{sec:cmatrix}) was the representation's ceiling, not
genus-2's.
(3)~\emph{Sharp-featured CAD benefits most}: fandisk's $10.4\times$
exceeds every analytic-family effect.
(4)~\emph{Baselines can pin like controls}: armadillo and horse read
the correct void count but pin its value at the same bound
($.0511$/$.0408$, seed-exact); the loss unpins both
($4.1\times$/$3.8\times$) at better Chamfer.
(5)~\emph{No headroom, no effect --- and no harm}: rocker-arm's one
measurable loop is already baseline-correct ($.0083$ vs the
$.04$--$.05$ band): \S\ref{sec:setup}'s imperfection criterion is unmet
and the loss stays inert ($0.96\times$, parity $1.00$, counts intact
everywhere) --- the method's honest boundary.
(6)~\emph{Count-level limits persist on sub-budget organic detail}:
spot's \emph{baseline} makes a phantom void in two of three seeds
($\#$sig $=2,1,2$); the loss clears two, retains one ($1,2,1$) ---
value repair reliable, count repair not guaranteed; its C1 tail pins at
$.0237$ (sd $0$), a milder floor of the double-torus kind.
\emph{And the pipeline reaches raw artist content}: the tom-yum pot
(Fig.~\ref{fig:matrix}; suppl.\ Fig.~S7), a raw Blender export
certified rather than constructed and bundled at $M{=}8192$ by the same
pre-registered rule (\S\ref{sec:setup}), cuts the void's value error
$3.9\times$ at Chamfer parity ($0.99\times$), the study's tightest seed
spread (C1 $.0056$--$.0059$ vs C0 $.0218$--$.0223$); the control
destroys the void ($\beta_2\,1{\to}0$, all seeds), pinned at the
unmatched bound ($.0223$) --- count, not value, is the verdict. Two caveats transfer: its baseline already reads the
correct count (a value-accuracy win, like bob), and ground truth is
certified, not constructed.
Group statistic (pre-registered): loops $1.52\times\pm0.55$ (2/3 pass),
voids $4.87\times\pm3.17$ (5/5); the passing wave spans
$1.5$--$10.4\times$ (suppl.\ \S G).

\section{Mechanism and discussion}
\label{sec:discussion}

\paragraph{A local signal cannot carry a global property --- sparse
tugs can.} A pixel's photometric gradient reaches a few vertices;
topology is global and categorical --- a wrong genus is not approached
smoothly; the blindness probes quantify the gap (\S\ref{sec:related})
--- no geometric proxy stands in. The persistence
term supplies the measurement with a usable gradient: each significant
pair back-propagates through the $3$--$7$ sample positions of its birth
and death simplices --- \emph{precise but sparse} pulls, iterated by re-pairing; the
resampling prior is the opposite --- broad, indirect, budget-limited. The results sort along this axis: loops respond to value
pressure (H1 win), wide shells to both channels (composition best),
components to neither (H0 null) --- and nothing responds where nothing
is broken (rocker-arm's no-headroom null, \S\ref{sec:generality}).

\paragraph{Why the control damages topology.} The norm-matched repulsion
imposes a length-scale floor on the sampled surface: the void's death
pins at that equilibrium (seed-identical) and drags while
near-diagonal phantom bars accumulate (suppl.\ Fig.~S1) and thin
structures collapse (cube $\beta_2$, torus $\beta_1$); the loss has no
length scale --- it pulls the death coordinate itself. The
same gradient budget through the same channel \emph{without}
topological information is actively harmful: C1's wins cannot be
extra-regularization artifacts.

\paragraph{Measurability at the loss's density.} A persistence loss can
only enforce features that clear the significance floor at its density
$M$: the torus void and two of the double torus's four loops do not, at
$M{=}2048$. We treat this as a principle, not a nuisance: the bundle is
built and audited at the same $M$, sub-floor features are excluded,
and the double torus is evaluated against its two measurable loops ---
all the soup can express at feasible budgets.

The floor is quantifiable: significance means lifetime
$> 6\,r_{\mathrm{med}}(M)$, and $r_{\mathrm{med}}$ follows the
uniform-sampling law $\propto M^{-1/2}$ (measured $.0090{\to}.0028$
over $M{=}2048{\to}20000$; ratio $3.16$ vs $\sqrt{9.77}{=}3.13$;
suppl.\ Fig.~S10). The
double torus's tube loops live at lifetime $.045$--$.046$ on the clean
surface: at $M{=}2048$ that is $0.83$--$0.85\times$ the floor --- invisible
by a $15\%$ margin, clearing it only from $M\!\approx\!3{\times}10^{3}$ ---
while at the evaluation density ($M{=}20000$, floor $.0171$) they clear
it $3.1\times$; the binding constraint is the \emph{representation}
(the soup never expresses the tubes): the \emph{measurement} floor
recedes as $M^{-1/2}$, the \emph{representation} floor is set by $N$
and no $M$ lowers it --- the double torus sits on the second. The torus quantifies the noise interaction: its second loop clears the
clean floor by only $1.33\times$ --- a margin live-cloud noise consumes,
leaving recruitment as its only gradient path (C6), the loss reaching
below its own floor (suppl.\ \S C).
Suppl.\ Table~S3 tabulates counts across $M$; the pot and the group
wave ran the rule \emph{prospectively} (\S\ref{sec:setup}; suppl.\
\S G), validated in-loop (Table~\ref{tab:gen}); adaptive density is the
natural extension.

\paragraph{Prescription.} Correct topology in the \emph{loss}; allocate
\emph{wide}; combine --- no schedule needed once the loss is calibrated
against the photometric gradient.

\section{Limitations}
\label{sec:limitations}

\textbf{Synthetic, single-machine.} All quantitative claims are on
synthetic renders with known topology: an
analytic family plus eight external genus-known meshes in two groups
(\S\ref{sec:generality}; one a certified raw artist mesh, the intended
scan path), drawn from public repositories, not a curated benchmark
suite. All surfaces
are closed except a first \emph{open-surface} probe, pre-registered
and passed: two bowls cut from the sphere (openings $0.14R$/$0.50R$,
certified disk topology) still read one significant void (the designed
rim's H1 never clears the floor); the loss cuts its error
$4.2\times$/$4.1\times$ at better-than-baseline Chamfer (ranges
disjoint); the control erases the wide bowl's void in 2/3 seeds
(suppl.\ \S F).
C7's Gaussian probe bounds noise
response, not scan realism (outliers, missing regions, registration
error). Real scans remain future work: their boundaries add
\emph{noise-born} H1 bars (unlike a designed rim); bar-filtering is the
near-term step.

\textbf{Density floor.} The loss cannot resolve features below the
significance floor at its density (torus void; two of four double-torus
loops at $M{=}2048$) --- audited, but a real ceiling. Count-level repair
of sub-budget detail is likewise out of scope (spot,
\S\ref{sec:generality}).

\textbf{Fixed target diagram.} The target bundle presumes known ground-truth
topology (a prior scan or template --- not blind capture); template-free
designs (total-persistence minimization, a target-\emph{genus} penalty,
class-level diagram priors) are untested.

\textbf{Control strength.} The repulsion control is aggressive (it
destroys topology at calibrated magnitude); a half-radius variant is
still $1.6\times$ worse (sphere $.0972$ vs $.0623$): verdicts do not
hinge on control tuning.

\textbf{Nondeterminism.} CUDA training is not bit-reproducible
(borderline prune/split decisions flip under float drift at
resampling); tail CVs run $7$--$10\%$ for baseline and loss arms alike,
a loss-identical re-run pair (sphere C6 $\equiv$ C1) landed $6\%$ apart
--- pure noise; statements follow \S\ref{sec:setup}'s reporting policy.

\textbf{Recruitment guarantees.} The zero-gradient pathology it repairs
is provable (Lemma~\ref{lem:zerograd}); its own self-correction premise
is empirical (suppl.\ \S C) --- pathological geometries could defeat
it; a failure-case study is future work.

\section{Conclusion}
\label{sec:conclusion}

We moved topology from the allocation channel into the objective: a
matched persistence loss with a Gabriel-\emph{certified} backward and a
load-bearing recruitment term, every gain beating a norm-matched
control at Chamfer parity. The hypotheses hold (RH1 objective root
cause; RH2 topology-specific; RH3 channels compose); verdicts
replicate across two external groups and degrade gracefully under
noise. Prescription: topology in the \emph{loss}, allocate \emph{wide},
combine; next: real scans.
\textbf{Reproducibility}: code, scenes, bundles, logs and seeds ship
on acceptance (nondeterminism: \S\ref{sec:limitations}).

{
  \small
  \bibliographystyle{ieeenat_fullname}
  \bibliography{refs}

@inproceedings{loper2014opendr,
  author    = {Loper, Matthew M. and Black, Michael J.},
  title     = {{OpenDR}: An Approximate Differentiable Renderer},
  booktitle = {European Conference on Computer Vision (ECCV)},
  pages     = {154--169},
  year      = {2014}
}

@inproceedings{kato2018neural,
  author    = {Kato, Hiroharu and Ushiku, Yoshitaka and Harada, Tatsuya},
  title     = {Neural {3D} Mesh Renderer},
  booktitle = {IEEE/CVF Conference on Computer Vision and Pattern Recognition (CVPR)},
  pages     = {3907--3916},
  year      = {2018}
}

@inproceedings{liu2019softras,
  author    = {Liu, Shichen and Li, Tianye and Chen, Weikai and Li, Hao},
  title     = {Soft Rasterizer: A Differentiable Renderer for Image-Based {3D} Reasoning},
  booktitle = {IEEE/CVF International Conference on Computer Vision (ICCV)},
  pages     = {7707--7716},
  year      = {2019}
}

@article{laine2020nvdiffrast,
  author  = {Laine, Samuli and Hellsten, Janne and Karras, Tero and Seol, Yeongho and Lehtinen, Jaakko and Aila, Timo},
  title   = {Modular Primitives for High-Performance Differentiable Rendering},
  journal = {ACM Transactions on Graphics},
  volume  = {39},
  number  = {6},
  pages   = {194:1--194:14},
  year    = {2020}
}

@inproceedings{tojo2026diffsoup,
  author    = {Tojo, Kenji and Bickel, Bernd and Umetani, Nobuyuki},
  title     = {{DiffSoup}: Direct Differentiable Rasterization of Triangle Soup for Extreme Radiance Field Simplification},
  booktitle = {IEEE/CVF Conference on Computer Vision and Pattern Recognition (CVPR)},
  year      = {2026},
  note      = {arXiv:2603.27151}
}

@article{kerbl2023gaussians,
  author  = {Kerbl, Bernhard and Kopanas, Georgios and Leimk{\"u}hler, Thomas and Drettakis, George},
  title   = {{3D} {Gaussian} Splatting for Real-Time Radiance Field Rendering},
  journal = {ACM Transactions on Graphics},
  volume  = {42},
  number  = {4},
  pages   = {139:1--139:14},
  year    = {2023}
}

@inproceedings{held2025trianglesplatting,
  author    = {Held, Jan and Vandeghen, Renaud and Deli{\`e}ge, Adrien and Hamdi, Abdullah and Giancola, Silvio and Rebain, Daniel and Cioppa, Anthony and Ghanem, Bernard and Vedaldi, Andrea and Tagliasacchi, Andrea and Van Droogenbroeck, Marc},
  title     = {Triangle Splatting for Real-Time Radiance Field Rendering},
  booktitle = {Proceedings of the International Conference on 3D Vision (3DV)},
  pages     = {1248--1257},
  year      = {2026},
  note      = {arXiv:2505.19175}
}

@inproceedings{held2026meshsplatting,
  author    = {Held, Jan and Son, Sanghyun and Vandeghen, Renaud and Rebain, Daniel and Gadelha, Matheus and Zhou, Yi and Cioppa, Anthony and Lin, Ming C. and Van Droogenbroeck, Marc and Tagliasacchi, Andrea},
  title     = {{MeshSplatting}: Differentiable Rendering with Opaque Meshes},
  booktitle = {IEEE/CVF Conference on Computer Vision and Pattern Recognition (CVPR)},
  year      = {2026},
  note      = {arXiv:2512.06818}
}

@article{nicolet2021largesteps,
  author  = {Nicolet, Baptiste and Jacobson, Alec and Jakob, Wenzel},
  title   = {Large Steps in Inverse Rendering of Geometry},
  journal = {ACM Transactions on Graphics},
  volume  = {40},
  number  = {6},
  year    = {2021}
}

@article{sharf2007topologyaware,
  author  = {Sharf, Andrei and Lewiner, Thomas and Shklarski, Gil and Toledo, Sivan and Cohen-Or, Daniel},
  title   = {Interactive Topology-Aware Surface Reconstruction},
  journal = {ACM Transactions on Graphics},
  volume  = {26},
  number  = {3},
  year    = {2007}
}

@article{ju2004repair,
  author  = {Ju, Tao},
  title   = {Robust Repair of Polygonal Models},
  journal = {ACM Transactions on Graphics},
  volume  = {23},
  number  = {3},
  pages   = {888--895},
  year    = {2004}
}

@article{attene2010repair,
  author  = {Attene, Marco},
  title   = {A Lightweight Approach to Repairing Digitized Polygon Meshes},
  journal = {The Visual Computer},
  volume  = {26},
  number  = {11},
  pages   = {1393--1406},
  year    = {2010}
}

@article{bruel2020toporecon,
  author  = {Br{\"u}el-Gabrielsson, Rickard and Ganapathi-Subramanian, Vignesh and Skraba, Primoz and Guibas, Leonidas J.},
  title   = {Topology-Aware Surface Reconstruction for Point Clouds},
  journal = {Computer Graphics Forum},
  volume  = {39},
  number  = {5},
  pages   = {197--207},
  year    = {2020}
}

@inproceedings{gabrielsson2020topologylayer,
  author    = {Br{\"u}el Gabrielsson, Rickard and Nelson, Bradley J. and Dwaraknath, Anjan and Skraba, Primoz},
  title     = {A Topology Layer for Machine Learning},
  booktitle = {International Conference on Artificial Intelligence and Statistics (AISTATS)},
  series    = {Proceedings of Machine Learning Research},
  volume    = {108},
  pages     = {1553--1563},
  year      = {2020}
}

@inproceedings{hu2019topopreserving,
  author    = {Hu, Xiaoling and Li, Fuxin and Samaras, Dimitris and Chen, Chao},
  title     = {Topology-Preserving Deep Image Segmentation},
  booktitle = {Advances in Neural Information Processing Systems 32 (NeurIPS)},
  year      = {2019}
}

@article{clough2022topoloss,
  author  = {Clough, James R. and Byrne, Nicholas and Oksuz, Ilkay and Zimmer, Veronika A. and Schnabel, Julia A. and King, Andrew P.},
  title   = {A Topological Loss Function for Deep-Learning Based Image Segmentation Using Persistent Homology},
  journal = {IEEE Transactions on Pattern Analysis and Machine Intelligence},
  volume  = {44},
  number  = {12},
  pages   = {8766--8778},
  year    = {2022}
}

@inproceedings{carriere2021optimizing,
  author    = {Carri{\`e}re, Mathieu and Chazal, Fr{\'e}d{\'e}ric and Glisse, Marc and Ike, Yuichi and Kannan, Hariprasad and Umeda, Yuhei},
  title     = {Optimizing Persistent Homology Based Functions},
  booktitle = {International Conference on Machine Learning (ICML)},
  series    = {Proceedings of Machine Learning Research},
  volume    = {139},
  year      = {2021}
}

@article{poulenard2018shapematching,
  author  = {Poulenard, Adrien and Skraba, Primoz and Ovsjanikov, Maks},
  title   = {Topological Function Optimization for Continuous Shape Matching},
  journal = {Computer Graphics Forum},
  volume  = {37},
  number  = {5},
  pages   = {13--25},
  year    = {2018}
}

@inproceedings{hoppe1996progressive,
  author    = {Hoppe, Hugues},
  title     = {Progressive Meshes},
  booktitle = {Proceedings of SIGGRAPH 96},
  pages     = {99--108},
  year      = {1996}
}

@inproceedings{garland1997qem,
  author    = {Garland, Michael and Heckbert, Paul S.},
  title     = {Surface Simplification Using Quadric Error Metrics},
  booktitle = {Proceedings of SIGGRAPH 97},
  pages     = {209--216},
  year      = {1997}
}

@article{alliez2003anisotropic,
  author  = {Alliez, Pierre and Cohen-Steiner, David and Devillers, Olivier and L{\'e}vy, Bruno and Desbrun, Mathieu},
  title   = {Anisotropic Polygonal Remeshing},
  journal = {ACM Transactions on Graphics},
  volume  = {22},
  number  = {3},
  pages   = {485--493},
  year    = {2003}
}

@inproceedings{botsch2004remeshing,
  author    = {Botsch, Mario and Kobbelt, Leif},
  title     = {A Remeshing Approach to Multiresolution Modeling},
  booktitle = {Eurographics/ACM SIGGRAPH Symposium on Geometry Processing (SGP)},
  year      = {2004}
}

@inproceedings{dunyach2013adaptive,
  author    = {Dunyach, Marion and Vanderhaeghe, David and Barthe, Lo{\"i}c and Botsch, Mario},
  title     = {Adaptive Remeshing for Real-Time Mesh Deformation},
  booktitle = {Eurographics 2013 --- Short Papers},
  year      = {2013}
}

@article{edelsbrunner2002persistence,
  author  = {Edelsbrunner, Herbert and Letscher, David and Zomorodian, Afra},
  title   = {Topological Persistence and Simplification},
  journal = {Discrete \& Computational Geometry},
  volume  = {28},
  number  = {4},
  pages   = {511--533},
  year    = {2002}
}

@article{cohensteiner2007stability,
  author  = {Cohen-Steiner, David and Edelsbrunner, Herbert and Harer, John},
  title   = {Stability of Persistence Diagrams},
  journal = {Discrete \& Computational Geometry},
  volume  = {37},
  number  = {1},
  pages   = {103--120},
  year    = {2007}
}

@article{edelsbrunner1994alphashapes,
  author  = {Edelsbrunner, Herbert and M{\"u}cke, Ernst P.},
  title   = {Three-Dimensional Alpha Shapes},
  journal = {ACM Transactions on Graphics},
  volume  = {13},
  number  = {1},
  pages   = {43--72},
  year    = {1994}
}

@inproceedings{maria2014gudhi,
  author    = {Maria, Cl{\'e}ment and Boissonnat, Jean-Daniel and Glisse, Marc and Yvinec, Mariette},
  title     = {The {Gudhi} Library: Simplicial Complexes and Persistent Homology},
  booktitle = {International Congress on Mathematical Software (ICMS)},
  series    = {Lecture Notes in Computer Science},
  volume    = {8592},
  pages     = {167--174},
  year      = {2014}
}

@article{leygonie2022framework,
  author  = {Leygonie, Jacob and Oudot, Steve and Tillmann, Ulrike},
  title   = {A Framework for Differential Calculus on Persistence Barcodes},
  journal = {Foundations of Computational Mathematics},
  volume  = {22},
  pages   = {1069--1131},
  year    = {2022}
}

@article{nigmetov2024bigsteps,
  author  = {Nigmetov, Arnur and Morozov, Dmitriy},
  title   = {Topological Optimization with Big Steps},
  journal = {Discrete \& Computational Geometry},
  year    = {2024},
  doi     = {10.1007/s00454-023-00613-x}
}

@inproceedings{mildenhall2020nerf,
  author    = {Mildenhall, Ben and Srinivasan, Pratul P. and Tancik, Matthew and Barron, Jonathan T. and Ramamoorthi, Ravi and Ng, Ren},
  title     = {{NeRF}: Representing Scenes as Neural Radiance Fields for View Synthesis},
  booktitle = {European Conference on Computer Vision (ECCV)},
  year      = {2020}
}

@inproceedings{wang2021neus,
  author    = {Wang, Peng and Liu, Lingjie and Liu, Yuan and Theobalt, Christian and Komura, Taku and Wang, Wenping},
  title     = {{NeuS}: Learning Neural Implicit Surfaces by Volume Rendering for Multi-view Reconstruction},
  booktitle = {Advances in Neural Information Processing Systems (NeurIPS)},
  year      = {2021}
}

@inproceedings{park2019deepsdf,
  author    = {Park, Jeong Joon and Florence, Peter and Straub, Julian and Newcombe, Richard and Lovegrove, Steven},
  title     = {{DeepSDF}: Learning Continuous Signed Distance Functions for Shape Representation},
  booktitle = {IEEE/CVF Conference on Computer Vision and Pattern Recognition (CVPR)},
  pages     = {165--174},
  year      = {2019}
}

@inproceedings{mescheder2019occupancy,
  author    = {Mescheder, Lars and Oechsle, Michael and Niemeyer, Michael and Nowozin, Sebastian and Geiger, Andreas},
  title     = {Occupancy Networks: Learning {3D} Reconstruction in Function Space},
  booktitle = {IEEE/CVF Conference on Computer Vision and Pattern Recognition (CVPR)},
  year      = {2019}
}

@inproceedings{li2023neuralangelo,
  author    = {Li, Zhaoshuo and M{\"u}ller, Thomas and Evans, Alex and Taylor, Russell H. and Unberath, Mathias and Liu, Ming-Yu and Lin, Chen-Hsuan},
  title     = {Neuralangelo: High-Fidelity Neural Surface Reconstruction},
  booktitle = {IEEE/CVF Conference on Computer Vision and Pattern Recognition (CVPR)},
  year      = {2023}
}

@inproceedings{wang2018pixel2mesh,
  author    = {Wang, Nanyang and Zhang, Yinda and Li, Zhuwen and Fu, Yanwei and Liu, Wei and Jiang, Yu-Gang},
  title     = {{Pixel2Mesh}: Generating {3D} Mesh Models from Single {RGB} Images},
  booktitle = {European Conference on Computer Vision (ECCV)},
  year      = {2018}
}

@article{hanocka2020point2mesh,
  author  = {Hanocka, Rana and Metzer, Gal and Giryes, Raja and Cohen-Or, Daniel},
  title   = {{Point2Mesh}: A Self-Prior for Deformable Meshes},
  journal = {ACM Transactions on Graphics},
  volume  = {39},
  number  = {4},
  year    = {2020}
}

@inproceedings{shen2021dmtet,
  author    = {Shen, Tianchang and Gao, Jun and Yin, Kangxue and Liu, Ming-Yu and Fidler, Sanja},
  title     = {Deep Marching Tetrahedra: a Hybrid Representation for High-Resolution {3D} Shape Synthesis},
  booktitle = {Advances in Neural Information Processing Systems (NeurIPS)},
  year      = {2021}
}

@inproceedings{munkberg2022nvdiffrec,
  author    = {Munkberg, Jacob and Hasselgren, Jon and Shen, Tianchang and Gao, Jun and Chen, Wenzheng and Evans, Alex and M{\"u}ller, Thomas and Fidler, Sanja},
  title     = {Extracting Triangular {3D} Models, Materials, and Lighting From Images},
  booktitle = {IEEE/CVF Conference on Computer Vision and Pattern Recognition (CVPR)},
  pages     = {8280--8290},
  year      = {2022}
}

@article{shen2023flexicubes,
  author  = {Shen, Tianchang and Munkberg, Jacob and Hasselgren, Jon and Yin, Kangxue and Wang, Zian and Chen, Wenzheng and Gojcic, Zan and Fidler, Sanja and Sharp, Nicholas and Gao, Jun},
  title   = {Flexible Isosurface Extraction for Gradient-Based Mesh Optimization},
  journal = {ACM Transactions on Graphics},
  volume  = {42},
  number  = {4},
  year    = {2023}
}

@inproceedings{peng2021sap,
  author    = {Peng, Songyou and Jiang, Chiyu and Liao, Yiyi and Niemeyer, Michael and Pollefeys, Marc and Geiger, Andreas},
  title     = {Shape As Points: A Differentiable {Poisson} Solver},
  booktitle = {Advances in Neural Information Processing Systems (NeurIPS)},
  year      = {2021}
}

@inproceedings{chen2019dibr,
  author    = {Chen, Wenzheng and Gao, Jun and Ling, Huan and Smith, Edward J. and Lehtinen, Jaakko and Jacobson, Alec and Fidler, Sanja},
  title     = {Learning to Predict {3D} Objects with an Interpolation-based Differentiable Renderer},
  booktitle = {Advances in Neural Information Processing Systems (NeurIPS)},
  year      = {2019}
}

@article{yifan2019dss,
  author  = {Yifan, Wang and Serena, Felice and Wu, Shihao and {\"O}ztireli, Cengiz and Sorkine-Hornung, Olga},
  title   = {Differentiable Surface Splatting for Point-based Geometry Processing},
  journal = {ACM Transactions on Graphics},
  volume  = {38},
  number  = {6},
  year    = {2019}
}

@inproceedings{groueix2018atlasnet,
  author    = {Groueix, Thibault and Fisher, Matthew and Kim, Vladimir G. and Russell, Bryan C. and Aubry, Mathieu},
  title     = {{AtlasNet}: A Papier-M{\^a}ch{\'e} Approach to Learning {3D} Surface Generation},
  booktitle = {IEEE/CVF Conference on Computer Vision and Pattern Recognition (CVPR)},
  year      = {2018}
}

@inproceedings{huang20242dgs,
  author    = {Huang, Binbin and Yu, Zehao and Chen, Anpei and Geiger, Andreas and Gao, Shenghua},
  title     = {{2D} {Gaussian} Splatting for Geometrically Accurate Radiance Fields},
  booktitle = {ACM SIGGRAPH 2024 Conference Papers},
  year      = {2024}
}

@inproceedings{guedon2024sugar,
  author    = {Gu{\'e}don, Antoine and Lepetit, Vincent},
  title     = {{SuGaR}: Surface-Aligned {Gaussian} Splatting for Efficient {3D} Mesh Reconstruction and High-Quality Mesh Rendering},
  booktitle = {IEEE/CVF Conference on Computer Vision and Pattern Recognition (CVPR)},
  pages     = {5354--5363},
  year      = {2024}
}

@inproceedings{shen2025topologygs,
  author    = {Shen, Tianqi and Liu, Shaohua and Feng, Jiaqi and Ma, Ziye and An, Ning},
  title     = {Topology-Aware {3D} {Gaussian} Splatting: Leveraging Persistent Homology for Optimized Structural Integrity},
  booktitle = {Proceedings of the AAAI Conference on Artificial Intelligence},
  volume    = {39},
  pages     = {6823--6832},
  year      = {2025}
}

@misc{gao2025genus,
  author       = {Gao, Xiang and Wang, Xinmu and Wu, Xiaolong and Li, Jiazhi and Shi, Jingyu and Guo, Yu and Liu, Yuanpeng and Song, Xiyun and Yu, Heather and Lin, Zongfang and Gu, Xianfeng David},
  title        = {Inverse Rendering for High-Genus Surface Meshes from Multi-View Images},
  howpublished = {arXiv:2511.18680},
  note         = {3DV 2026},
  year         = {2025}
}

@misc{gao2026homology,
  author       = {Gao, Xiang and Wang, Xinmu and Liu, Yuanpeng and Wang, Yue and Huang, Junqi and Chen, Wei and Gu, Xianfeng},
  title        = {Inverse Rendering for High-Genus {3D} Surface Meshes from Multi-view Images with Persistent Homology Priors},
  howpublished = {arXiv:2601.12155},
  note         = {ICASSP 2026},
  year         = {2026}
}

@misc{jignasu2024stitch,
  author       = {Jignasu, Anushrut and Herron, Ethan and Jiang, Zhanhong and Sarkar, Soumik and Hegde, Chinmay and Ganapathysubramanian, Baskar and Balu, Aditya and Krishnamurthy, Adarsh},
  title        = {{STITCH}: Surface Reconstruction using Implicit Neural Representations with Topology Constraints and Persistent Homology},
  howpublished = {arXiv:2412.18696},
  year         = {2024}
}

@inproceedings{waibel2022topological,
  author    = {Waibel, Dominik J. E. and Atwell, Scott and Meier, Matthias and Marr, Carsten and Rieck, Bastian},
  title     = {Capturing Shape Information with Multi-Scale Topological Loss Terms for {3D} Reconstruction},
  booktitle = {Medical Image Computing and Computer Assisted Intervention (MICCAI)},
  pages     = {150--159},
  year      = {2022}
}

@misc{burgdorfer2025radiant,
  author       = {Burgdorfer, Nathaniel and Mordohai, Philippos},
  title        = {Radiant Triangle Soup with Soft Connectivity Forces for {3D} Reconstruction and Novel View Synthesis},
  howpublished = {arXiv:2505.23642},
  year         = {2025}
}

@inproceedings{gropp2020igr,
  author    = {Gropp, Amos and Yariv, Lior and Haim, Niv and Atzmon, Matan and Lipman, Yaron},
  title     = {Implicit Geometric Regularization for Learning Shapes},
  booktitle = {International Conference on Machine Learning (ICML)},
  year      = {2020}
}

@inproceedings{yariv2020idr,
  author    = {Yariv, Lior and Kasten, Yoni and Moran, Dror and Galun, Meirav and Atzmon, Matan and Basri, Ronen and Lipman, Yaron},
  title     = {Multiview Neural Surface Reconstruction by Disentangling Geometry and Appearance},
  booktitle = {Advances in Neural Information Processing Systems (NeurIPS)},
  year      = {2020}
}

@inproceedings{peng2020convoccnet,
  author    = {Peng, Songyou and Niemeyer, Michael and Mescheder, Lars and Pollefeys, Marc and Geiger, Andreas},
  title     = {Convolutional Occupancy Networks},
  booktitle = {European Conference on Computer Vision (ECCV)},
  year      = {2020}
}

@inproceedings{oechsle2021unisurf,
  author    = {Oechsle, Michael and Peng, Songyou and Geiger, Andreas},
  title     = {{UNISURF}: Unifying Neural Implicit Surfaces and Radiance Fields for Multi-View Reconstruction},
  booktitle = {IEEE/CVF International Conference on Computer Vision (ICCV)},
  year      = {2021}
}

@inproceedings{yariv2021volsdf,
  author    = {Yariv, Lior and Gu, Jiatao and Kasten, Yoni and Lipman, Yaron},
  title     = {Volume Rendering of Neural Implicit Surfaces},
  booktitle = {Advances in Neural Information Processing Systems (NeurIPS)},
  year      = {2021}
}

@article{muller2022instantngp,
  author  = {M{\"u}ller, Thomas and Evans, Alex and Schied, Christoph and Keller, Alexander},
  title   = {Instant Neural Graphics Primitives with a Multiresolution Hash Encoding},
  journal = {ACM Transactions on Graphics},
  volume  = {41},
  number  = {4},
  year    = {2022}
}

@inproceedings{yu2022monosdf,
  author    = {Yu, Zehao and Peng, Songyou and Niemeyer, Michael and Sattler, Torsten and Geiger, Andreas},
  title     = {{MonoSDF}: Exploring Monocular Geometric Cues for Neural Implicit Surface Reconstruction},
  booktitle = {Advances in Neural Information Processing Systems (NeurIPS)},
  year      = {2022}
}

@inproceedings{wang2023neus2,
  author    = {Wang, Yiming and Han, Qin and Habermann, Marc and Daniilidis, Kostas and Theobalt, Christian and Liu, Lingjie},
  title     = {{NeuS2}: Fast Learning of Neural Implicit Surfaces for Multi-view Reconstruction},
  booktitle = {IEEE/CVF International Conference on Computer Vision (ICCV)},
  year      = {2023}
}

@inproceedings{yariv2023bakedsdf,
  author    = {Yariv, Lior and Hedman, Peter and Reiser, Christian and Verbin, Dor and Srinivasan, Pratul P. and Szeliski, Richard and Barron, Jonathan T. and Mildenhall, Ben},
  title     = {{BakedSDF}: Meshing Neural {SDFs} for Real-Time View Synthesis},
  booktitle = {ACM SIGGRAPH 2023 Conference Proceedings},
  year      = {2023}
}

@inproceedings{huang2023nksr,
  author    = {Huang, Jiahui and Gojcic, Zan and Atzmon, Matan and Litany, Or and Fidler, Sanja and Williams, Francis},
  title     = {Neural Kernel Surface Reconstruction},
  booktitle = {IEEE/CVF Conference on Computer Vision and Pattern Recognition (CVPR)},
  year      = {2023}
}

@inproceedings{lorensen1987marching,
  author    = {Lorensen, William E. and Cline, Harvey E.},
  title     = {Marching Cubes: A High Resolution {3D} Surface Construction Algorithm},
  booktitle = {Proceedings of the 14th Annual Conference on Computer Graphics and Interactive Techniques (SIGGRAPH)},
  pages     = {163--169},
  year      = {1987}
}

@inproceedings{ju2002dual,
  author    = {Ju, Tao and Losasso, Frank and Schaefer, Scott and Warren, Joe},
  title     = {Dual Contouring of {Hermite} Data},
  booktitle = {Proceedings of the 29th Annual Conference on Computer Graphics and Interactive Techniques (SIGGRAPH)},
  pages     = {339--346},
  year      = {2002}
}

@inproceedings{kazhdan2006poisson,
  author    = {Kazhdan, Michael and Bolitho, Matthew and Hoppe, Hugues},
  title     = {{Poisson} Surface Reconstruction},
  booktitle = {Proceedings of the Fourth Eurographics Symposium on Geometry Processing (SGP)},
  pages     = {61--70},
  year      = {2006}
}

@article{kazhdan2013screened,
  author  = {Kazhdan, Michael and Hoppe, Hugues},
  title   = {Screened {Poisson} Surface Reconstruction},
  journal = {ACM Transactions on Graphics},
  volume  = {32},
  number  = {3},
  year    = {2013}
}

@inproceedings{remelli2020meshsdf,
  author    = {Remelli, Edoardo and Lukoianov, Artem and Richter, Stephan R. and Guillard, Beno{\^i}t and Bagautdinov, Timur and Baque, Pierre and Fua, Pascal},
  title     = {{MeshSDF}: Differentiable Iso-Surface Extraction},
  booktitle = {Advances in Neural Information Processing Systems (NeurIPS)},
  year      = {2020}
}

@inproceedings{gao2020deftet,
  author    = {Gao, Jun and Chen, Wenzheng and Xiang, Tommy and {Fuji Tsang}, Clement and Jacobson, Alec and McGuire, Morgan and Fidler, Sanja},
  title     = {Learning Deformable Tetrahedral Meshes for {3D} Reconstruction},
  booktitle = {Advances in Neural Information Processing Systems (NeurIPS)},
  year      = {2020}
}

@article{chen2021nmc,
  author  = {Chen, Zhiqin and Zhang, Hao},
  title   = {Neural Marching Cubes},
  journal = {ACM Transactions on Graphics},
  volume  = {40},
  number  = {6},
  year    = {2021}
}

@inproceedings{chen2020bspnet,
  author    = {Chen, Zhiqin and Tagliasacchi, Andrea and Zhang, Hao},
  title     = {{BSP-Net}: Generating Compact Meshes via Binary Space Partitioning},
  booktitle = {IEEE/CVF Conference on Computer Vision and Pattern Recognition (CVPR)},
  year      = {2020}
}

@inproceedings{gkioxari2019meshrcnn,
  author    = {Gkioxari, Georgia and Malik, Jitendra and Johnson, Justin},
  title     = {Mesh {R-CNN}},
  booktitle = {IEEE/CVF International Conference on Computer Vision (ICCV)},
  pages     = {9784--9794},
  year      = {2019}
}

@article{hu2018tetwild,
  author  = {Hu, Yixin and Zhou, Qingnan and Gao, Xifeng and Jacobson, Alec and Zorin, Denis and Panozzo, Daniele},
  title   = {Tetrahedral Meshing in the Wild},
  journal = {ACM Transactions on Graphics},
  volume  = {37},
  number  = {4},
  year    = {2018}
}

@inproceedings{zwicker2001splatting,
  author    = {Zwicker, Matthias and Pfister, Hanspeter and van Baar, Jeroen and Gross, Markus},
  title     = {Surface Splatting},
  booktitle = {Proceedings of the 28th Annual Conference on Computer Graphics and Interactive Techniques (SIGGRAPH)},
  pages     = {371--378},
  year      = {2001}
}

@inproceedings{aliev2020npbg,
  author    = {Aliev, Kara-Ali and Sevastopolsky, Artem and Kolos, Maria and Ulyanov, Dmitry and Lempitsky, Victor},
  title     = {Neural Point-Based Graphics},
  booktitle = {European Conference on Computer Vision (ECCV)},
  pages     = {696--712},
  year      = {2020}
}

@inproceedings{lassner2021pulsar,
  author    = {Lassner, Christoph and Zollh{\"o}fer, Michael},
  title     = {Pulsar: Efficient Sphere-Based Neural Rendering},
  booktitle = {IEEE/CVF Conference on Computer Vision and Pattern Recognition (CVPR)},
  pages     = {1440--1449},
  year      = {2021}
}

@article{ruckert2022adop,
  author  = {R{\"u}ckert, Darius and Franke, Linus and Stamminger, Marc},
  title   = {{ADOP}: Approximate Differentiable One-Pixel Point Rendering},
  journal = {ACM Transactions on Graphics},
  volume  = {41},
  number  = {4},
  year    = {2022}
}

@inproceedings{xu2022pointnerf,
  author    = {Xu, Qiangeng and Xu, Zexiang and Philip, Julien and Bi, Sai and Shu, Zhixin and Sunkavalli, Kalyan and Neumann, Ulrich},
  title     = {Point-{NeRF}: Point-Based Neural Radiance Fields},
  booktitle = {IEEE/CVF Conference on Computer Vision and Pattern Recognition (CVPR)},
  pages     = {5438--5448},
  year      = {2022}
}

@inproceedings{yu2024mipsplatting,
  author    = {Yu, Zehao and Chen, Anpei and Huang, Binbin and Sattler, Torsten and Geiger, Andreas},
  title     = {Mip-Splatting: Alias-Free {3D} {Gaussian} Splatting},
  booktitle = {IEEE/CVF Conference on Computer Vision and Pattern Recognition (CVPR)},
  year      = {2024}
}

@inproceedings{fang2024minisplatting,
  author    = {Fang, Guangchi and Wang, Bing},
  title     = {Mini-Splatting: Representing Scenes with a Constrained Number of {Gaussians}},
  booktitle = {European Conference on Computer Vision (ECCV)},
  year      = {2024}
}

@inproceedings{dai2024surfels,
  author    = {Dai, Pinxuan and Xu, Jiamin and Xie, Wenxiang and Liu, Xinguo and Wang, Huamin and Xu, Weiwei},
  title     = {High-Quality Surface Reconstruction Using {Gaussian} Surfels},
  booktitle = {ACM SIGGRAPH 2024 Conference Papers},
  year      = {2024}
}

@article{yu2024gof,
  author  = {Yu, Zehao and Sattler, Torsten and Geiger, Andreas},
  title   = {{Gaussian} Opacity Fields: Efficient Adaptive Surface Reconstruction in Unbounded Scenes},
  journal = {ACM Transactions on Graphics},
  volume  = {43},
  number  = {6},
  year    = {2024}
}

@article{chen2024pgsr,
  author  = {Chen, Danpeng and Li, Hai and Ye, Weicai and Wang, Yifan and Xie, Weijian and Zhai, Shangjin and Wang, Nan and Liu, Haomin and Bao, Hujun and Zhang, Guofeng},
  title   = {{PGSR}: Planar-Based {Gaussian} Splatting for Efficient and High-Fidelity Surface Reconstruction},
  journal = {IEEE Transactions on Visualization and Computer Graphics},
  year    = {2024}
}

@article{guedon2025milo,
  author  = {Gu{\'e}don, Antoine and Gomez, Diego and Maruani, Nissim and Gong, Bingchen and Drettakis, George and Ovsjanikov, Maks},
  title   = {{MILo}: Mesh-in-the-Loop {Gaussian} Splatting for Detailed and Efficient Surface Reconstruction},
  journal = {ACM Transactions on Graphics},
  volume  = {44},
  number  = {6},
  year    = {2025}
}

@article{zomorodian2005computing,
  author  = {Zomorodian, Afra and Carlsson, Gunnar},
  title   = {Computing Persistent Homology},
  journal = {Discrete \& Computational Geometry},
  volume  = {33},
  pages   = {249--274},
  year    = {2005}
}

@article{bauer2021ripser,
  author  = {Bauer, Ulrich},
  title   = {{Ripser}: Efficient Computation of {Vietoris--Rips} Persistence Barcodes},
  journal = {Journal of Applied and Computational Topology},
  volume  = {5},
  pages   = {391--423},
  year    = {2021}
}

@article{kerber2017geometry,
  author  = {Kerber, Michael and Morozov, Dmitriy and Nigmetov, Arnur},
  title   = {Geometry Helps to Compare Persistence Diagrams},
  journal = {ACM Journal of Experimental Algorithmics},
  volume  = {22},
  year    = {2017}
}

@article{adams2017images,
  author  = {Adams, Henry and Emerson, Tegan and Kirby, Michael and Neville, Rachel and Peterson, Chris and Shipman, Patrick and Chepushtanova, Sofya and Hanson, Eric and Motta, Francis and Ziegelmeier, Lori},
  title   = {Persistence Images: A Stable Vector Representation of Persistent Homology},
  journal = {Journal of Machine Learning Research},
  volume  = {18},
  number  = {8},
  pages   = {1--35},
  year    = {2017}
}

@article{bubenik2015landscapes,
  author  = {Bubenik, Peter},
  title   = {Statistical Topological Data Analysis using Persistence Landscapes},
  journal = {Journal of Machine Learning Research},
  volume  = {16},
  number  = {3},
  pages   = {77--102},
  year    = {2015}
}

@inproceedings{hofer2017signatures,
  author    = {Hofer, Christoph and Kwitt, Roland and Niethammer, Marc and Uhl, Andreas},
  title     = {Deep Learning with Topological Signatures},
  booktitle = {Advances in Neural Information Processing Systems (NIPS)},
  year      = {2017}
}

@inproceedings{hofer2020filtration,
  author    = {Hofer, Christoph and Graf, Florian and Rieck, Bastian and Niethammer, Marc and Kwitt, Roland},
  title     = {Graph Filtration Learning},
  booktitle = {International Conference on Machine Learning (ICML)},
  pages     = {4314--4323},
  year      = {2020}
}

@inproceedings{carriere2020perslay,
  author    = {Carri{\`e}re, Mathieu and Chazal, Fr{\'e}d{\'e}ric and Ike, Yuichi and Lacombe, Th{\'e}o and Royer, Martin and Umeda, Yuhei},
  title     = {{PersLay}: A Neural Network Layer for Persistence Diagrams and New Graph Topological Signatures},
  booktitle = {International Conference on Artificial Intelligence and Statistics (AISTATS)},
  pages     = {2786--2796},
  year      = {2020}
}

@inproceedings{moor2020topoae,
  author    = {Moor, Michael and Horn, Max and Rieck, Bastian and Borgwardt, Karsten},
  title     = {Topological Autoencoders},
  booktitle = {International Conference on Machine Learning (ICML)},
  pages     = {7045--7054},
  year      = {2020}
}

@inproceedings{shit2021cldice,
  author    = {Shit, Suprosanna and Paetzold, Johannes C. and Sekuboyina, Anjany and Ezhov, Ivan and Unger, Alexander and Zhylka, Andrey and Pluim, Josien P. W. and Bauer, Ulrich and Menze, Bjoern H.},
  title     = {{clDice} -- A Novel Topology-Preserving Loss Function for Tubular Structure Segmentation},
  booktitle = {IEEE/CVF Conference on Computer Vision and Pattern Recognition (CVPR)},
  year      = {2021}
}

@inproceedings{stucki2023betti,
  author    = {Stucki, Nico and Paetzold, Johannes C. and Shit, Suprosanna and Menze, Bjoern and Bauer, Ulrich},
  title     = {Topologically Faithful Image Segmentation via Induced Matching of Persistence Barcodes},
  booktitle = {International Conference on Machine Learning (ICML)},
  pages     = {32698--32727},
  year      = {2023}
}

@article{chazal2021introduction,
  author  = {Chazal, Fr{\'e}d{\'e}ric and Michel, Bertrand},
  title   = {An Introduction to Topological Data Analysis: Fundamental and Practical Aspects for Data Scientists},
  journal = {Frontiers in Artificial Intelligence},
  volume  = {4},
  pages   = {667963},
  year    = {2021}
}

@misc{carriere2026survey,
  author       = {Carri{\`e}re, Mathieu and Ike, Yuichi and Lacombe, Th{\'e}o and Nishikawa, Naoki},
  title        = {Persistence-Based Topological Optimization: A Survey},
  howpublished = {arXiv:2603.24613},
  year         = {2026}
}

@inproceedings{guskov2001topological,
  author    = {Guskov, Igor and Wood, Zo{\"e}},
  title     = {Topological Noise Removal},
  booktitle = {Proceedings of Graphics Interface},
  pages     = {19--26},
  year      = {2001}
}

@article{wood2004excess,
  author  = {Wood, Zo{\"e} and Hoppe, Hugues and Desbrun, Mathieu and Schr{\"o}der, Peter},
  title   = {Removing Excess Topology from Isosurfaces},
  journal = {ACM Transactions on Graphics},
  volume  = {23},
  number  = {2},
  pages   = {190--208},
  year    = {2004}
}

@article{attene2013repairing,
  author  = {Attene, Marco and Campen, Marcel and Kobbelt, Leif},
  title   = {Polygon Mesh Repairing: An Application Perspective},
  journal = {ACM Computing Surveys},
  volume  = {45},
  number  = {2},
  year    = {2013}
}

@misc{huang2020manifoldplus,
  author       = {Huang, Jingwei and Zhou, Yichao and Guibas, Leonidas},
  title        = {{ManifoldPlus}: A Robust and Scalable Watertight Manifold Surface Generation Method for Triangle Soups},
  howpublished = {arXiv:2005.11621},
  year         = {2020}
}

@misc{zhou2016thingi10k,
  author       = {Zhou, Qingnan and Jacobson, Alec},
  title        = {{Thingi10K}: A Dataset of 10,000 {3D}-Printing Models},
  howpublished = {arXiv:1605.04797},
  year         = {2016}
}

@inproceedings{schonberger2016sfm,
  author    = {Sch{\"o}nberger, Johannes L. and Frahm, Jan-Michael},
  title     = {Structure-from-Motion Revisited},
  booktitle = {IEEE Conference on Computer Vision and Pattern Recognition (CVPR)},
  pages     = {4104--4113},
  year      = {2016}
}

@inproceedings{schonberger2016pixelwise,
  author    = {Sch{\"o}nberger, Johannes L. and Zheng, Enliang and Frahm, Jan-Michael and Pollefeys, Marc},
  title     = {Pixelwise View Selection for Unstructured Multi-View Stereo},
  booktitle = {European Conference on Computer Vision (ECCV)},
  pages     = {501--518},
  year      = {2016}
}

@article{furukawa2010pmvs,
  author  = {Furukawa, Yasutaka and Ponce, Jean},
  title   = {Accurate, Dense, and Robust Multiview Stereopsis},
  journal = {IEEE Transactions on Pattern Analysis and Machine Intelligence},
  volume  = {32},
  number  = {8},
  pages   = {1362--1376},
  year    = {2010}
}

@inproceedings{jensen2014dtu,
  author    = {Jensen, Rasmus and Dahl, Anders and Vogiatzis, George and Tola, Engin and Aan{\ae}s, Henrik},
  title     = {Large Scale Multi-view Stereopsis Evaluation},
  booktitle = {IEEE Conference on Computer Vision and Pattern Recognition (CVPR)},
  pages     = {406--413},
  year      = {2014}
}

@article{knapitsch2017tanks,
  author  = {Knapitsch, Arno and Park, Jaesik and Zhou, Qian-Yi and Koltun, Vladlen},
  title   = {Tanks and Temples: Benchmarking Large-Scale Scene Reconstruction},
  journal = {ACM Transactions on Graphics},
  volume  = {36},
  number  = {4},
  year    = {2017}
}

@misc{hu2024topodiffusion,
  author       = {Hu, Jiangbei and Fei, Ben and Xu, Baixin and Hou, Fei and Yang, Weidong and Wang, Shengfa and Lei, Na and Qian, Chen and He, Ying},
  title        = {Topology-Aware Latent Diffusion for {3D} Shape Generation},
  howpublished = {arXiv:2401.17603},
  year         = {2024}
}

@inproceedings{roell2024dect,
  author    = {Roell, Ernst and Rieck, Bastian},
  title     = {Differentiable Euler Characteristic Transforms for Shape Classification},
  booktitle = {International Conference on Learning Representations (ICLR)},
  year      = {2024}
}
}

\clearpage
\appendix
\setcounter{figure}{0}
\setcounter{table}{0}
\renewcommand{\thetable}{S\arabic{table}}
\renewcommand{\thefigure}{S\arabic{figure}}

\twocolumn[{%
  \centering
  \vspace*{1em}
  {\Large\bf Supplementary Material}\\[6pt]
  \normalsize Appendices A--H. The main text refers to these sections as
  ``suppl.\ \S A''\,--\,``suppl.\ \S H'' and to their floats as
  ``Table~S1'', ``Fig.~S1'', and so on.\\
  \vspace*{1.5em}
}]

\section{The allocation-channel study, in brief}
\label{app:companion}

This appendix reports, with their numbers, the three findings of our
allocation-channel study --- controlled resampling-prior experiments on the
same pipeline, scenes, and budgets, conducted before moving topology into
the objective --- that this paper builds on.

\paragraph{(1) Geometric metrics are topology-blind.} Three constructed
cases each pit a topologically correct candidate A against a wrong one B
whose geometry is bisection-matched, making Chamfer \emph{equal by
construction} (40k samples, metrics normalized by the ground-truth bounding
diagonal; Table~\ref{tab:blindness}). The discriminating-dimension
bottleneck separates every pair --- by $35$--$40\times$ in the H0/H2 cases
--- while in the first two cases Hausdorff$_{95}$ actually \emph{prefers}
the topologically wrong candidate. Chamfer cannot tell a benign bump from a
destroyed void; the persistence metric can.

\begin{table*}[t]
\centering
\caption{Blindness probes: A topologically correct, B wrong, geometry
bisection-matched so Chamfer is equal by construction. Bottleneck to ground
truth in the discriminating dimension.}
\label{tab:blindness}
\small
\begin{tabular}{lcccc}
\toprule
case (disc.\ dim) & Cham.\% A $=$ B & bott.\ A & bott.\ B & ratio \\
\midrule
thin handle, genus $0{\to}1$ (H1) & 0.916 & ${\approx}0$ & .0302 & --- \\
bridge merges two components (H0) & 0.628 & .0014 & .0574 & $40\times$ \\
punctured enclosed void (H2)      & 0.977 & .0040 & .1385 & $35\times$ \\
\bottomrule
\end{tabular}
\end{table*}

\paragraph{(2) The prior arms, and why voids are ``width-primary''.} The
study biases DiffSoup's own resampler with a precomputed importance field
built from the target's significant features. Each significant feature $f$
(lifetime $p_f$ above the significance threshold) contributes Gaussian
kernels at the coordinates $C_f$ of its birth/death-simplex vertices ---
the loop-filling triangle, void-filling tetrahedron, or component-merging
edge --- sized by the feature's own death scale:
\begin{equation}
\begin{gathered}
\phi_d(q)=\!\!\sum_{f:\,\dim f=d}\,\sum_{c\in C_f}\! w_f\,
  e^{-\lVert q-c\rVert^2/(2\sigma_f^2)},\qquad w_f = p_f/s^2,\\
\sigma_f = s\,\mathrm{clip}\!\big(\sqrt{\smash[b]{\mathrm{death}_f}},\,
  3r_{\mathrm{med}},\,\tfrac{1}{2}\mathrm{diag}\big),
\end{gathered}
\end{equation}
evaluated on a dense target-surface sample, normalized per dimension to
$[0,1]$ (99.5th percentile), combined by max; a query takes its nearest
surface sample's value. The spread knob $s$ widens each kernel while
preserving its injected surface mass ($\sigma\,{\to}\,s\sigma$,
$w\,{\to}\,w/s^2$, so $w\sigma^2$ is $s$-invariant): $s{=}1$ is the
concentrated prior (B2), $s{=}3$ the spread prior (B4) --- the field C5
reuses. The field enters resampling at strict budget parity through two
levers: prune protection --- faces drop in order of $\mathrm{vis} +
\lambda\,Q_{0.9}(\mathrm{vis})\cdot\mathrm{imp}$, removing exactly the
required count --- and respawn bias --- top-importance visible faces
subdivide first, then the protective order re-prunes to the cap. Remaining
arms: B0 baseline; B1 = B2 applied at initialization only; B3/B5
random-center non-topological controls width-matched to B2/B4 respectively.
On the void class
(Table~\ref{tab:companion}; $N{=}1200$, five paired seeds) B4 cuts the
bottleneck $33$--$53\%$ below baseline --- but the width-matched random
control B5 recovers most of that gain, leaving only a small, shape-dependent
topological residual (cylinder $4.6\sigma$, sphere $2.2\sigma$, cube tie);
one-shot placement washes out entirely (B1 $\approx$ B0, $.0546$ vs $.0535$
on the sphere). Hence the verdict quoted in the main paper's introduction:
the prior's topological value is carried largely by its spatial
\emph{width}.

\begin{table}[t]
\centering
\caption{Allocation study, void class (H2), $N{=}1200$, five paired seeds:
tail bottleneck to target (lower is better). Spread beats concentrated;
the width-matched random control (B5) recovers most of the spread gain.}
\label{tab:companion}
\footnotesize
\setlength{\tabcolsep}{3pt}
\begin{tabular}{lccccc}
\toprule
 & & B2 & B3 & B4 & B5 \\
shape & B0 & conc. & narrow rand. & spread & wide rand. \\
\midrule
sphere   & .0535 & .0438 & .0440 & \textbf{.0357} & .0374 \\
cube     & .0579 & .0409 & .0403 & .0273 & \textbf{.0270} \\
cylinder & .0545 & .0440 & .0477 & \textbf{.0360} & .0402 \\
\bottomrule
\end{tabular}
\end{table}

\paragraph{(3) No prior shape repairs loops.} On the torus ($N{=}700$, five
paired seeds, true $\beta_1{=}2$) every arm fails the topology-specificity
test: the best condition is the \emph{non-topological} wide control B5
($.0267$), ahead of spread-topological B4 ($.0316$), baseline ($.0409$), and
narrow random B3 ($.0397$); the concentrated topological field B2 is worst
($.0425$) \emph{and} manufactures phantom handles --- final significant-H1
count $4.4$ against the true $2.0$, with the worst Chamfer ($1.32$ vs
baseline $1.01$). Spreading merely stops the harm; it never beats the random
control. That failure is the opening premise of the main paper, and the
class the loss channel wins there (main paper \S5.1, Table~1).

\section{Diagnostics: density sensitivity and diagram trajectories}
\label{app:diagnostics}

\paragraph{Sampling density vs.\ measurability.}
Table~\ref{tab:msens} counts the significant bars of the clean target
surfaces across sampling densities under the self-calibrating floor
$6\,r_{\mathrm{med}}(M)$ of the main paper's \S6. It confirms the three
density claims made in the paper: the sphere's void clears the floor at
every tested density; the torus needs $M{=}2048$ before its \emph{second}
loop is measurable, and its own void becomes measurable only at $M{=}4096$
(hence the H1 restriction at $M{=}2048$); and the double torus's tube
loops cross the floor between $M{=}2048$ and $M{=}4096$ --- consistent
with the $M\gtrsim3{\times}10^{3}$ estimate there ---
so all four loops are measurable at $M{=}4096$.

\begin{table}[t]
\centering\small
\caption{Significant bars (H1 / H2) on the clean target surfaces vs.\
sampling density $M$; floor $=6\,r_{\mathrm{med}}(M)$, seed 0.}
\label{tab:msens}
\setlength{\tabcolsep}{4pt}
\begin{tabular}{lcccc}
\toprule
 & $M{=}512$ & $M{=}1024$ & $M{=}2048$ & $M{=}4096$ \\
\midrule
sphere       & 0 / 1 & 0 / 1 & 0 / 1 & 0 / 1 \\
torus        & 1 / 0 & 1 / 0 & 2 / 0 & 2 / 1 \\
double torus & 0 / 0 & 2 / 0 & 2 / 0 & 4 / 0 \\
\bottomrule
\end{tabular}
\end{table}

\paragraph{Diagram trajectories: the pinning mechanism, visualized.}
Fig.~\ref{fig:pdtraj} shows the live H2 diagrams of one sphere seed under
C0/C1/C2 at four training steps. All three arms start with the void
essentially on target (initialization samples the SfM point cloud), and
photometric training alone drags it away (the baseline settles
${\approx}.06$ below). The loss arm dips identically while the ramp is
inactive, then pulls the void back to within ${\approx}.016$ of target ---
matching its measured tail. The control arm is qualitatively different:
the void's death is dragged progressively down (${\approx}.23 \to .15 \to
.13$) while a carpet of phantom mid-persistence bars accumulates near the
diagonal --- the repulsion-equilibrium signature discussed in the main
paper's \S6.

\begin{figure}[tb]
\centering
\includegraphics[width=\linewidth]{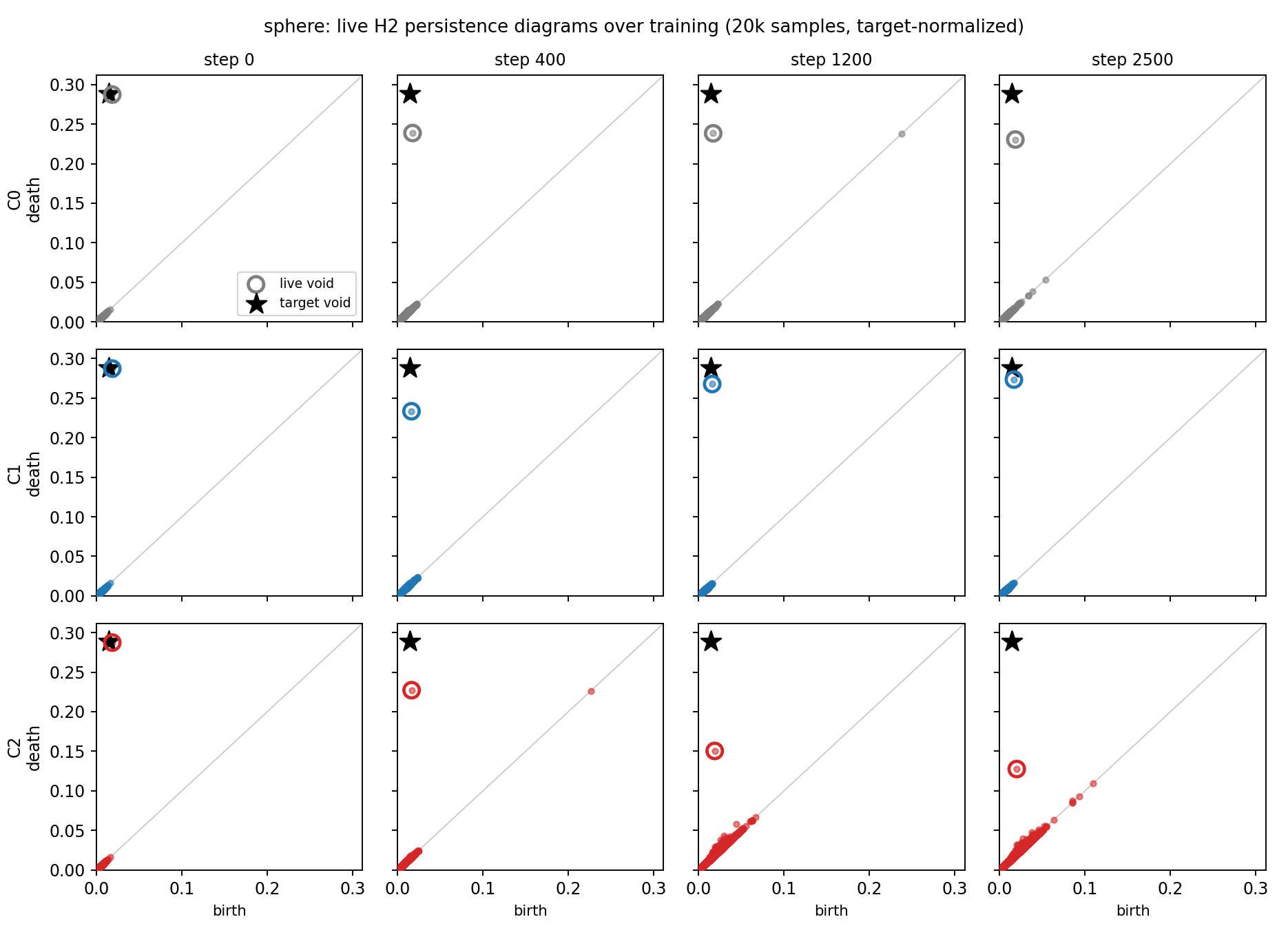}
\caption{Live H2 persistence diagrams over training (one sphere seed; 20k
samples, target-normalized; star $=$ target void bar, circle $=$ most
persistent live bar). \emph{Top}, baseline: the void drifts off target.
\emph{Middle}, loss: the void returns to the target once the ramp
activates. \emph{Bottom}, repulsion control: the void is pinned and
dragged away while near-diagonal phantom bars accumulate.}
\label{fig:pdtraj}
\end{figure}

\paragraph{The genus-2 stretch case, in full.}
On the double torus every arm --- baseline, loss, control --- returns a
bottleneck of exactly $.0264$ (sd $0$ over five seeds; raised from two
in revision, identical values throughout), and Wasserstein is
equally pinned ($.0527$--$.0528$): the two tube loops are unexpressable
at feasible budgets (replicating the allocation study's floor, \S A) and
dominate both diagram metrics regardless of guidance. The loss adds no
phantoms ($\#$sig $\beta_1$ $=2$ every run) at marginally best Chamfer.

\paragraph{Sensor-noise internals (C7/C7h).}
Recruitment absorbs the injected noise: at $\sigma{=}1\%$ the torus's
second loop is sub-threshold at every refresh, with zero unreached
targets. Calibration is the noisiest stage ($\lambda_{\mathrm{peak}}$
spreads ${\approx}3\times$ across seeds), yet tails stay tight ---
consistent with the flat $\rho$ response reported in the main paper.

\paragraph{The tom-yum pot: challenge inventory and prospective floor
protocol.}
The generality wave's artist mesh (main paper \S4, \S5.2) compounds the
failure modes real scans will bring. Table~\ref{tab:potchal} inventories
them with the measured consequence of each --- every row is backed by
the ingest certificate or by a measurement in this study, none is
asserted. Table~\ref{tab:potstair} records the staircase behind the flue
row: the certified shell's significant-bar signature across sampling
density, which fixed the observable class (H2) and the bundle density
($M{=}8192$) \emph{before} any training run.

\begin{table}[tb]
\centering
\caption{Tom-yum pot challenge inventory: what makes the shape hard, and
the measured consequence in this study.}
\label{tab:potchal}
\footnotesize
\setlength{\tabcolsep}{3pt}
\renewcommand{\arraystretch}{1.15}
\begin{tabular}{p{0.315\linewidth}p{0.615\linewidth}}
\toprule
challenge & measured consequence \\
\midrule
raw artist export ($6{,}318$ vertices): 9 open shells, mixed quad/12-gon
faces, 50 boundary
edges after exact weld, 232 non-manifold junction edges & no trustworthy
ground-truth topology as exported; ingest-and-certify required (weld,
offset-solidify into a thin shell via an exact unsigned distance field
$+$ marching cubes at iso $\varepsilon$, drop 804 enclosed pocket
shells) \\
certified shell: one body, genus 3, $\beta{=}(1,6,1)$ exact,
$190{,}990$ vertices & ground truth known by \emph{certificate} (exact
simplicial homology; watertightness cross-check; bit-identical
rebuilds), not by construction \\
six thin handle/vent loops & every H1 loop sits below the significance
floor at every working density --- H1 is unusable as the loss
observable \\
narrow flue mouth, capping at small $\alpha$ & the flue interior reads
as an enclosed on-axis chamber --- an H2 void, first measurable at
$M{=}8192$ (margin $1.20$--$1.23\times$, five sampling seeds); the
floor rule picks class and density prospectively
(Table~\ref{tab:potstair}) \\
in-loop outcome & C1 cuts the void's error $3.9\times$ at Chamfer
parity; the repulsion control \emph{destroys} it ($\beta_2\,1{\to}0$,
all seeds, $1.36\times$ worse Chamfer) --- main paper Table~2 \\
\bottomrule
\end{tabular}
\end{table}

\begin{table}[tb]
\centering
\caption{Tom-yum pot staircase (seed 0; the $M{=}8192$ margin spans
seeds 0--4): significant-bar signature
$(\beta_0,\beta_1,\beta_2)$ of the certified shell vs.\ sampling density
$M$, floor $=6\,r_{\mathrm{med}}(M)$. The bundle is built at the first
density where the discriminating feature clears ($M{=}8192$: exactly one
significant bar, H2, stable to re-sampling at ${\sim}10^{-5}$).}
\label{tab:potstair}
\footnotesize
\setlength{\tabcolsep}{5pt}
\begin{tabular}{lcccc}
\toprule
 & $M{=}2048$ & $M{=}4096$ & $M{=}8192$ & $M{=}20000$ \\
\midrule
sig.\ bars & $(1,0,0)$ & $(1,0,0)$ & $(1,0,1)$ & $(1,0,1)$ \\
\bottomrule
\end{tabular}
\end{table}

\section{Additional result figures and the repair gate}
\label{app:extrafigs}

Fig.~\ref{fig:gen} shows the generality-wave training trajectories
(main paper \S5.2, Table~2); Fig.~\ref{fig:tails} the per-seed tail
spread behind the sphere row of main-paper Table~1;
Fig.~\ref{fig:nsig} the torus significant-count check behind
\S5.1's zero-phantom-handles claim.

\begin{figure*}[t]
\centering
\includegraphics[width=0.495\linewidth]{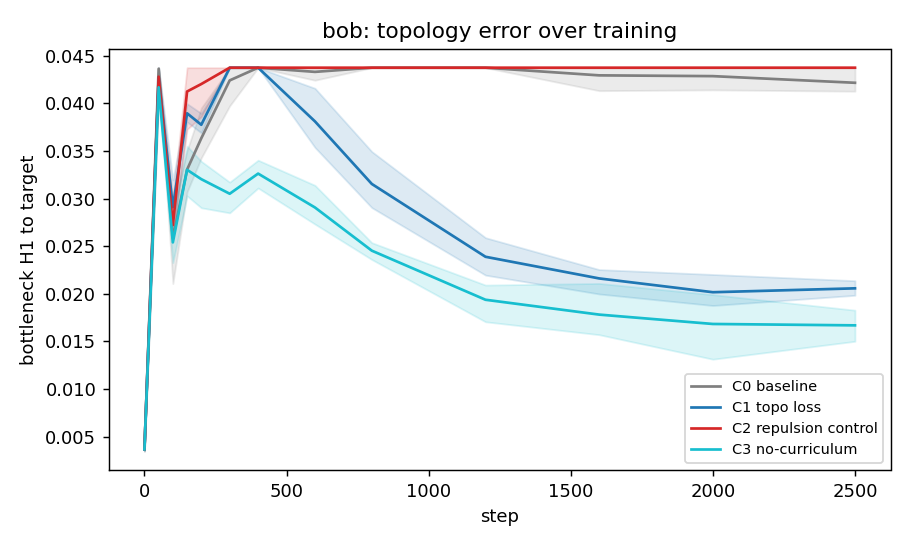}\hfill
\includegraphics[width=0.495\linewidth]{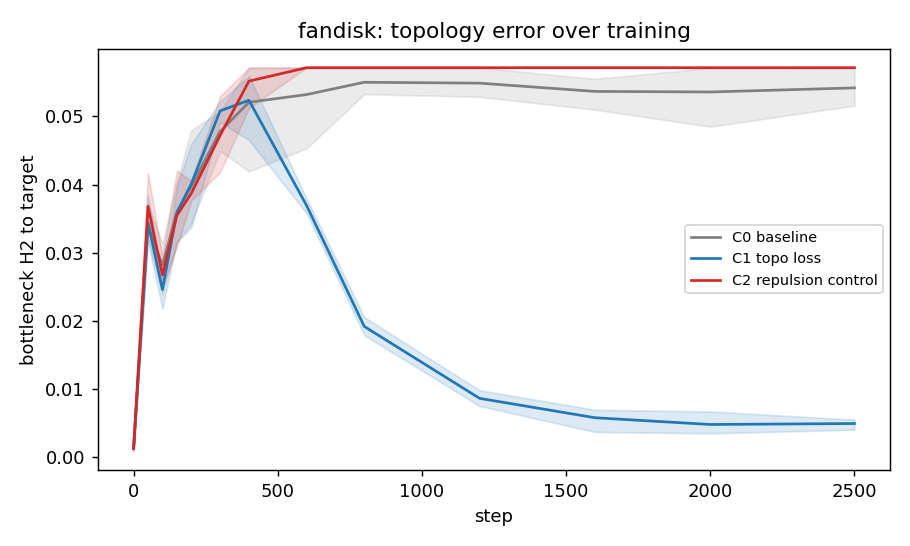}
\caption{Generality trajectories (line = seed mean, band = seed min--max).
\emph{Left}: bob (genus 1) --- the H1 result off the analytic family: the
loss halves the loop error with $\#$sig H1 $= 2$ in every seed, while the
control rides the baseline after collapsing a loop. \emph{Right}: fandisk
(CAD) --- the largest effect measured in this study ($10.4\times$).}
\label{fig:gen}
\end{figure*}

\begin{figure}[t]
\centering
\includegraphics[width=0.9\linewidth]{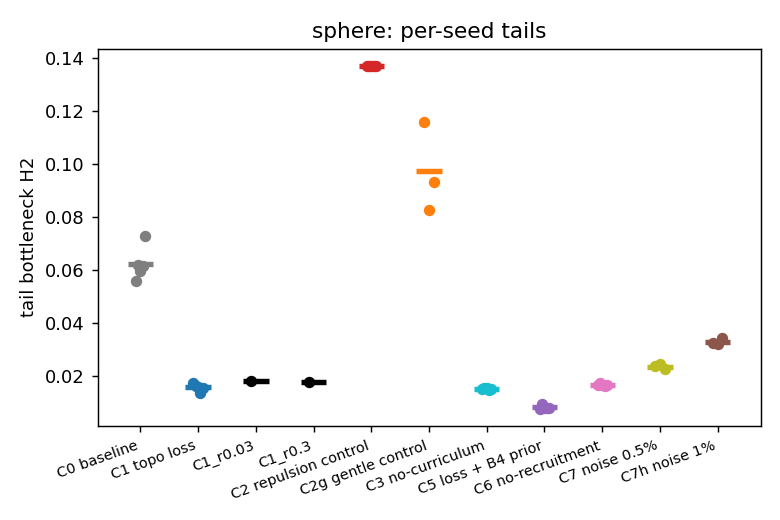}
\caption{Sphere, per-seed tail bottlenecks (dots) with arm means (bars).
Every loss arm sits far below baseline with tight seed spread; the two
control strengths bracket baseline from \emph{above}; the $\rho$ ablation
arms are indistinguishable from $\rho{=}0.1$.}
\label{fig:tails}
\end{figure}

\begin{figure}[t]
\centering
\includegraphics[width=0.85\linewidth]{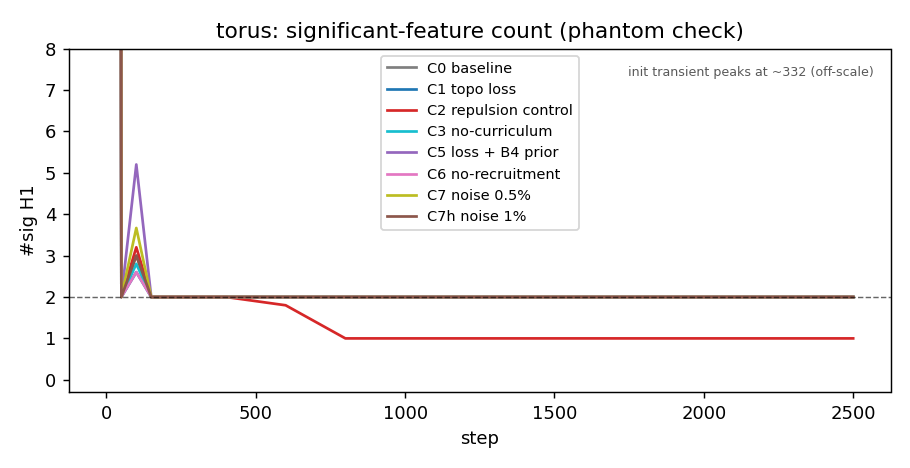}
\caption{Torus phantom check (main paper \S5.1): significant-H1 count over
training (seed means; dashed line: true $\beta_1{=}2$). After the brief
resampling transient near step 100, every loss arm sits exactly on two
significant loops --- zero phantom handles --- while the repulsion control
(C2) collapses a loop by step ${\sim}800$ and never recovers it.}
\label{fig:nsig}
\end{figure}

\paragraph{The stage-3a gate (main paper \S5).} Fig.~\ref{fig:toys} shows
three of the four repair probes run before integration: the loss alone
(no photometric term, Adam directly on point positions) must repair
defective clouds. The predicted location-blind pathology of recruitment
did not materialize, for a mechanism reason worth stating: tearing a
2-manifold shell cannot grow a component bar past the local detour scale
(the tug is abandoned), while cutting the 1-D bridge raises the recruited
bar's death monotonically toward target --- the runaway concentrates on
the only profitable direction.

\begin{figure}[t]
\centering
\includegraphics[width=\linewidth]{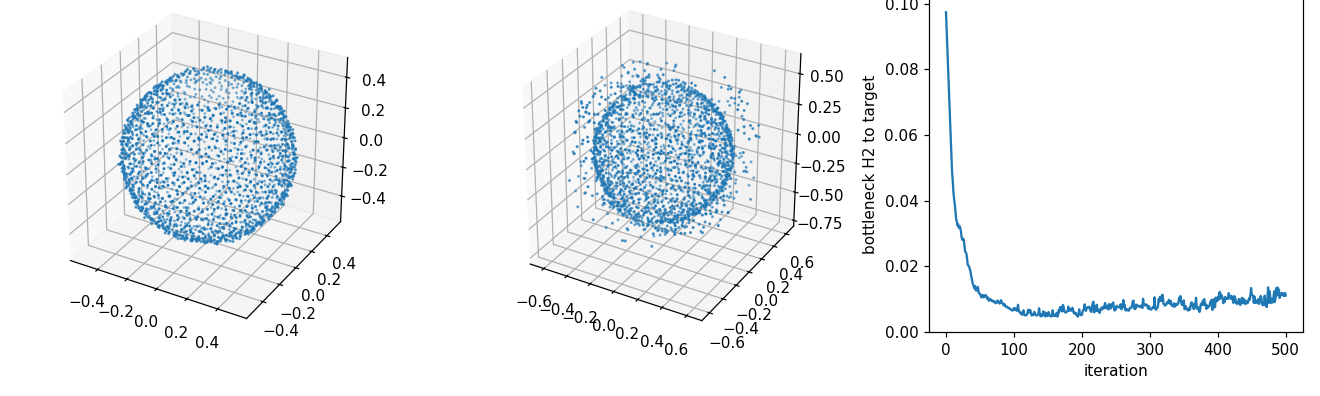}\\[2pt]
\includegraphics[width=\linewidth]{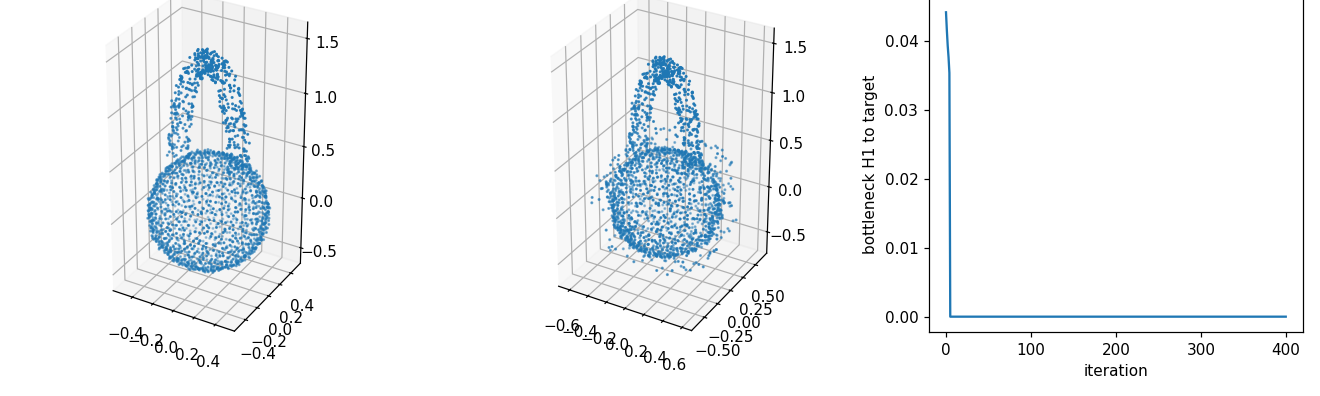}\\[2pt]
\includegraphics[width=\linewidth]{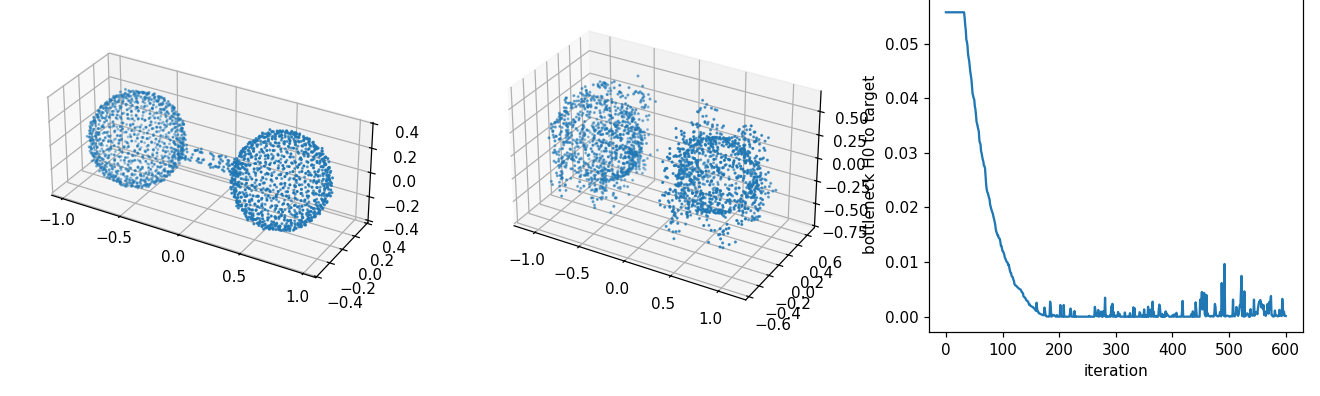}
\caption{The stage-3a gate: the loss alone (Adam on raw points, no
photometric term) repairs defective clouds; each row shows initial cloud,
final cloud, and bottleneck-to-target trajectory. \emph{Top} (T1):
punctured sphere --- the matched birth term closes the void ($8.8\times$).
\emph{Middle} (T2): spurious handle --- the diagonal term kills the extra
H1 bar. \emph{Bottom} (T4): bridged merge with no significant live H0 bar
--- recruitment severs the bridge ($450\times$, $\beta_0\,1{\to}2$), the
case where plain optimal matching provably has zero gradient. (T3,
mis-spaced components, $1229\times$: omitted for space.)}
\label{fig:toys}
\end{figure}

\begin{figure*}[t]
\centering
\includegraphics[width=0.495\linewidth]{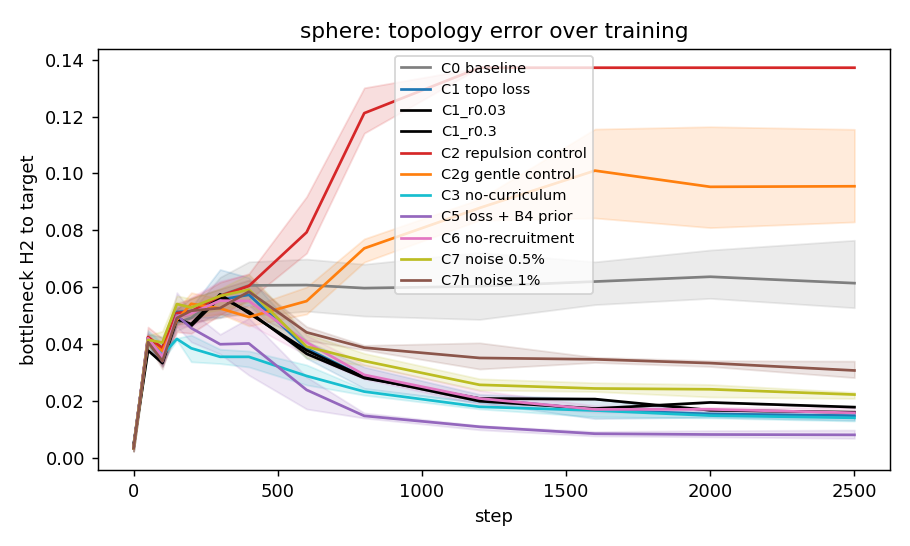}\hfill
\includegraphics[width=0.495\linewidth]{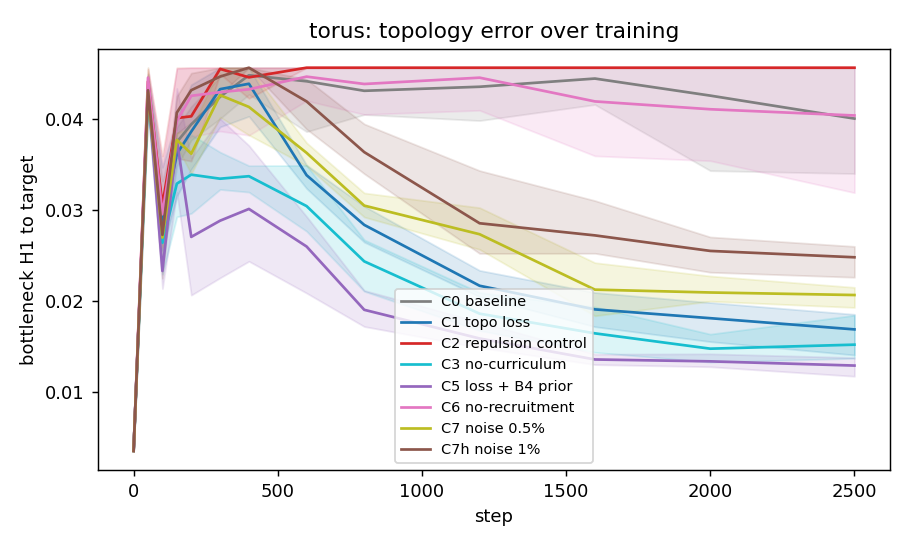}
\caption{Topology error over training: bottleneck to target in the
discriminating dimension (line = seed mean, band = seed min--max).
\emph{Left}, sphere (H2): every loss arm (C1, C3, C5, C6, both $\rho$
variants, noise pair C7/C7h) converges to a fraction of baseline; the
norm-matched repulsion control C2 is \emph{destructive}, and its gentle
variant C2g still ends worse than baseline. \emph{Right}, torus (H1), the
class no prior won: same ordering, C5 best; noise arms descend shallower
but never approach baseline; the no-recruitment ablation C6 rides it.}
\label{fig:mainseries}
\end{figure*}

\begin{figure}[t]
\centering
\includegraphics[width=0.365\linewidth]{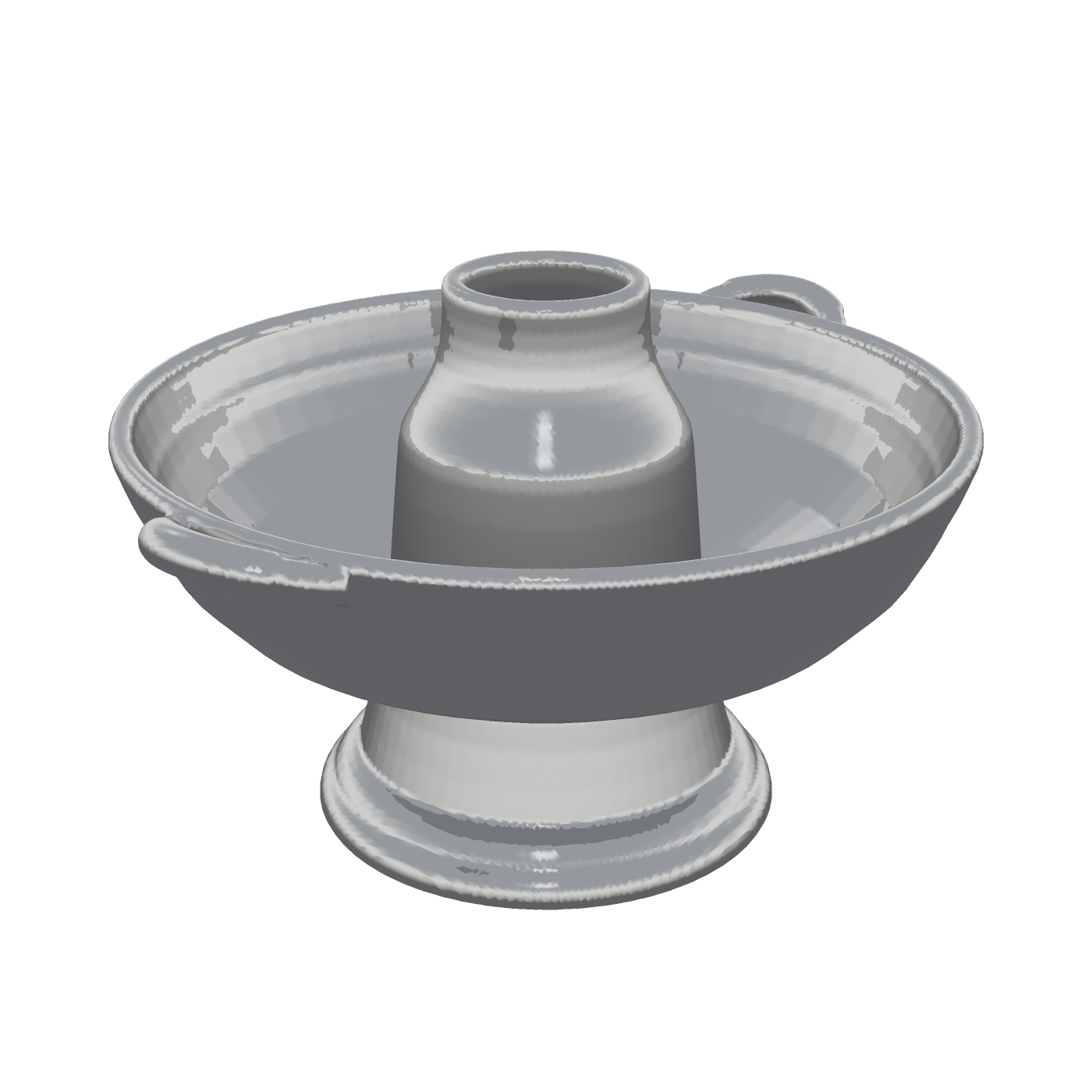}\hfill
\includegraphics[width=0.365\linewidth]{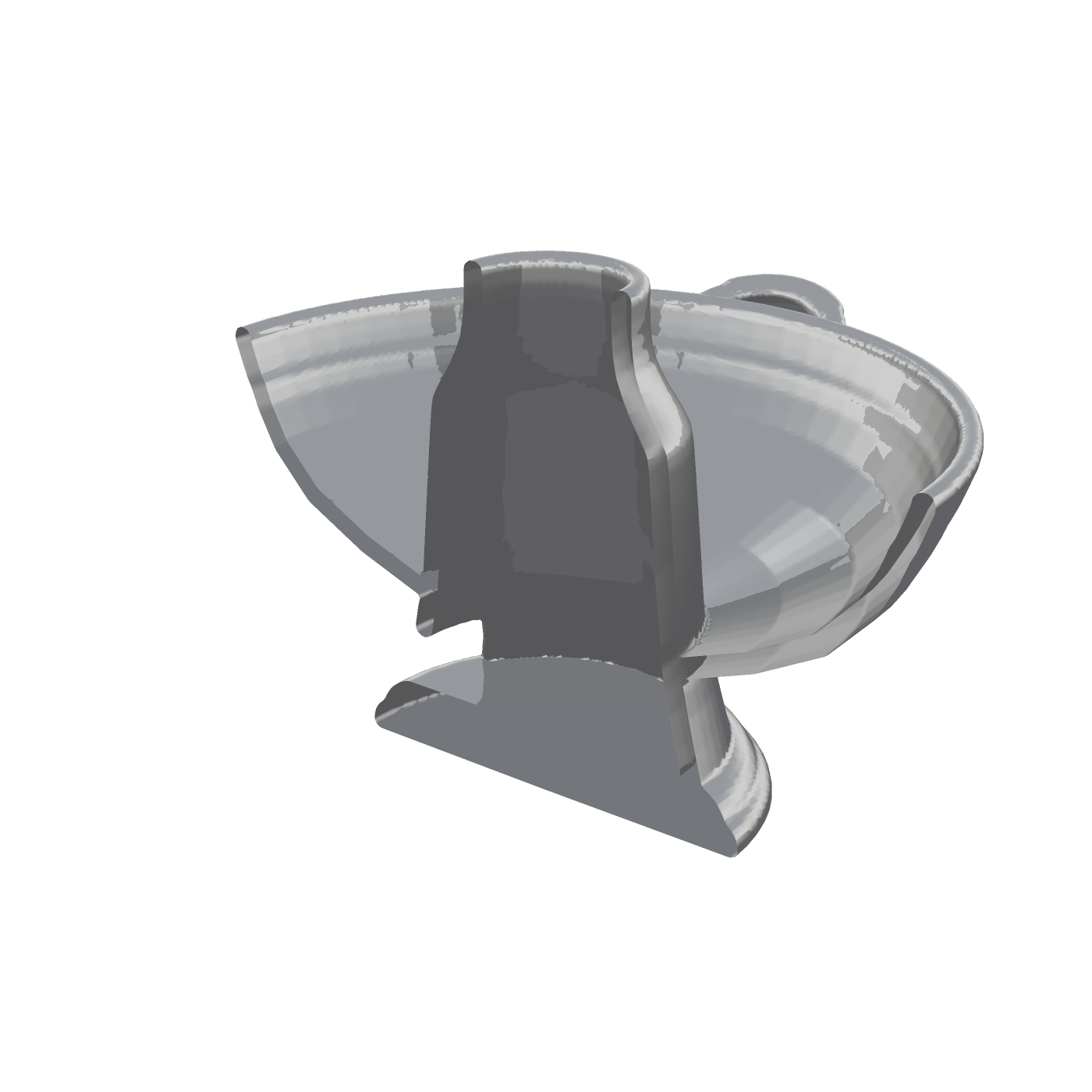}
\caption{The \emph{tom-yum pot}: an artist-authored Blender model (raw
export: 9 open shells, 232 non-manifold junction edges) ingested and
certified to a single genus-3 shell, $\beta{=}(1,6,1)$ exact
(main paper \S4). \emph{Right}: cutaway --- the narrow flue mouth caps
at small $\alpha$, so the flue interior reads as an enclosed chamber: the
on-axis H2 void the floor rule selects at $M{=}8192$ (all H1 loops sit
below the floor).}
\label{fig:tomyum}
\end{figure}

\begin{figure*}[t]
\centering
\includegraphics[width=0.9\linewidth]{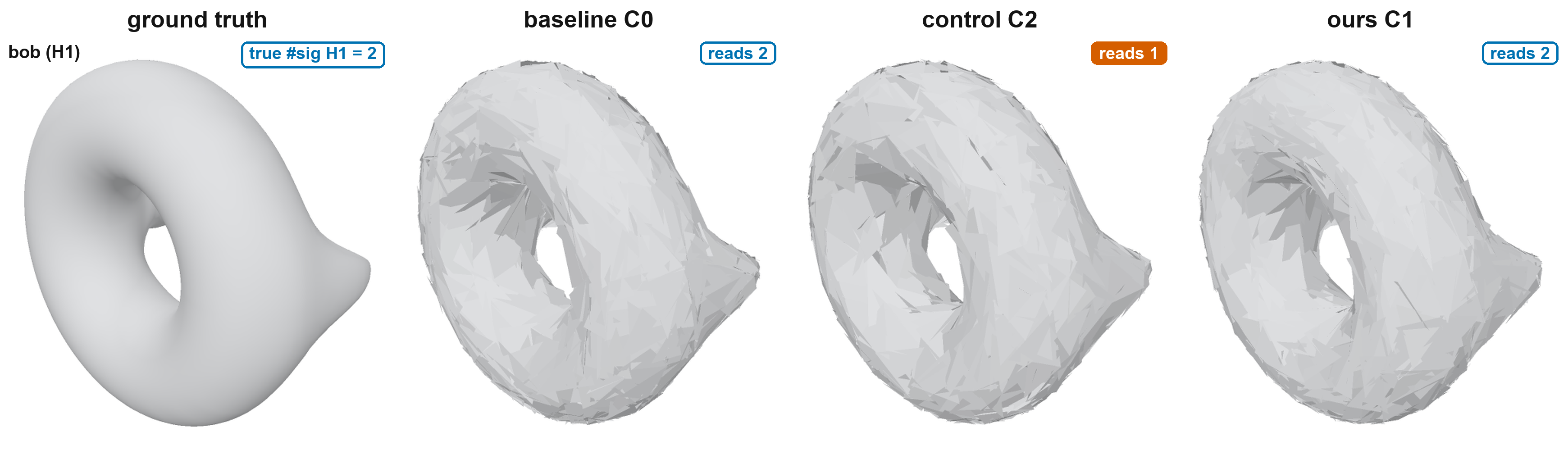}
\caption{The loop-class row of the main paper's Fig.~1 matrix: bob
(genus 1, H1). The repulsion control's reconstruction reads one
significant loop against the true two ($\beta_1\,2{\to}1$, all seeds)
at $2\times$ worse Chamfer, while baseline and loss both read two ---
the loss at half the baseline's tail error (main Table~2).}
\label{fig:bobmatrix}
\end{figure*}

\begin{figure*}[t]
\centering
\includegraphics[width=0.92\linewidth]{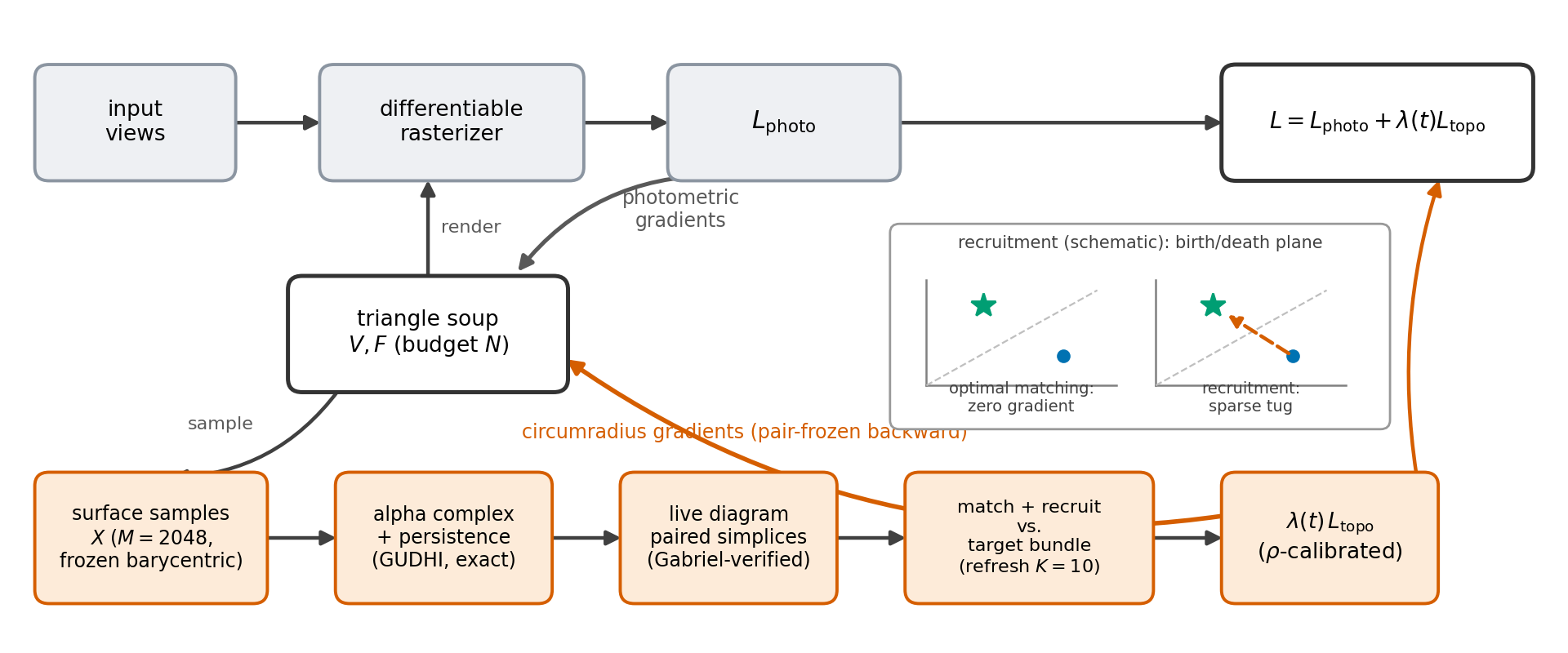}
\caption{The system, schematically. DiffSoup's photometric path (gray)
is untouched; the topological branch (orange) samples the live soup
($M{=}2048$, frozen barycentric weights), computes exact alpha-complex
persistence, matches and recruits against the fixed target bundle
(refresh every $K{=}10$ steps and after every resample), and
back-propagates through the pair-frozen circumradius re-expressions
into the same vertices (main \S3). \emph{Inset} (schematic): an absent
target bar (star) receives no gradient under optimal matching (main
Lemma~1); recruitment attaches the nearest unclaimed live bar --- the
sparse tug.}
\label{fig:pipeline}
\end{figure*}

\begin{figure}[t]
\centering
\includegraphics[width=\linewidth]{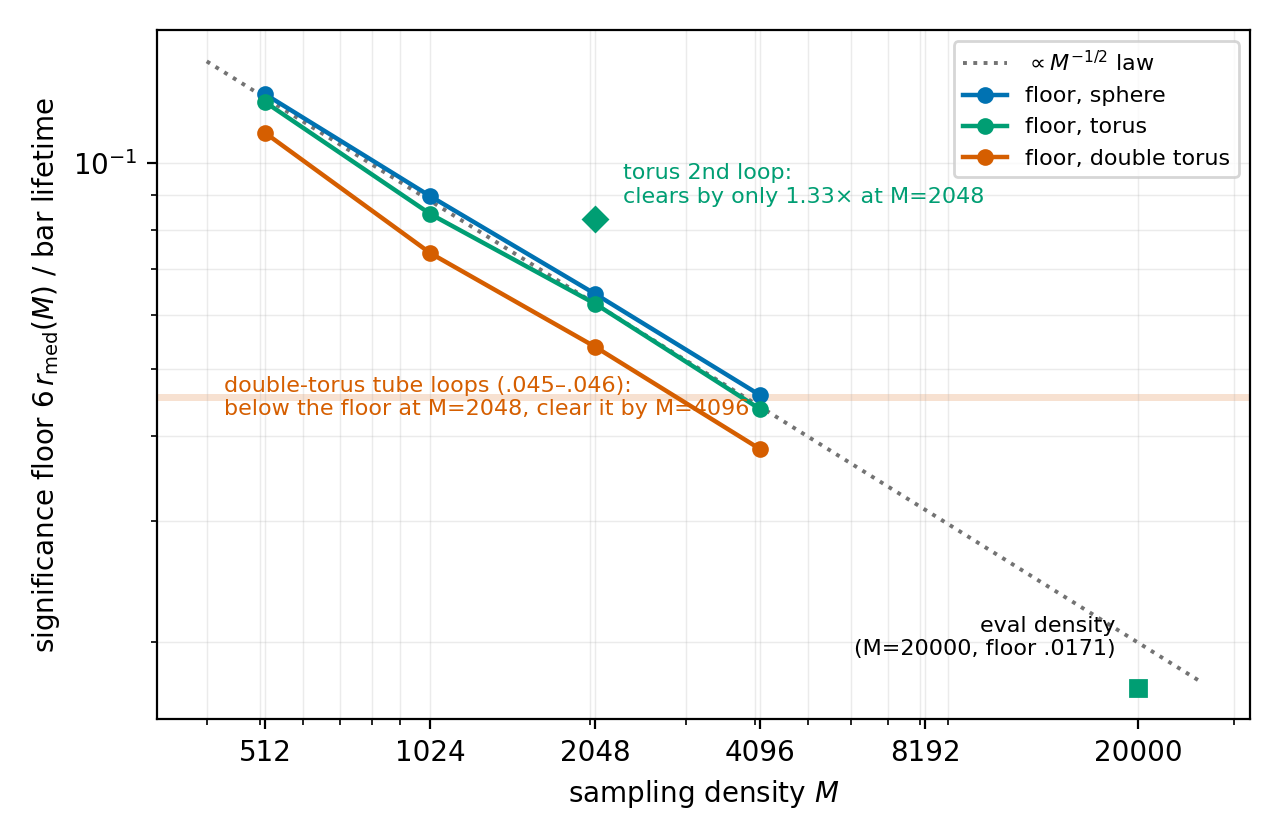}
\caption{The measurement floor, visualized (all values from the
recorded density sweep behind Table~\ref{tab:msens} and main \S6). The
floor $6\,r_{\mathrm{med}}(M)$ falls as $M^{-1/2}$ (dotted: the law
anchored at the torus $M{=}2048$ point, reaching the $M{=}20000$
evaluation floor $.0171$), while a feature's lifetime is fixed --- a
feature becomes measurable where the falling floor crosses it. The
double torus's tube loops ($.045$--$.046$ band) sit below the floor at
$M{=}2048$ and clear it by $M{=}4096$, matching Table~\ref{tab:msens}'s
counts; the torus's second loop clears by only $1.33\times$ at
$M{=}2048$ (lifetime marked at $1.33\times$ the $.0623$ floor) --- the
thin margin that live-cloud noise consumes (main \S6, C6).}
\label{fig:floor}
\end{figure}

\section{A decade of reconstruction, cut by optimization philosophy}
\label{app:survey}

The main paper's related work argues one sentence that this appendix
substantiates in breadth: reconstruction methods of the differentiable
decade, whatever their representation, optimize \emph{geometric}
objectives (photometric fidelity plus geometric regularizers) or steer
\emph{sampling} under such an objective --- topology, when correct, is
correct by construction, by extraction, or by accident, never because
the training objective measured it. Table~\ref{tab:philosophy}
tabulates representative methods under this cut, including the nearest
concurrent topology-aware entries; the subsections survey each family
and the adjacent literatures this paper draws machinery or discipline
from.

\begin{table*}[t]
\centering
\caption{Representative reconstruction methods by optimization
philosophy. ``Topology in objective?'' asks whether the \emph{training
objective} measures a topological quantity of the reconstructed
\emph{surface} (not whether the representation constrains topology, and
not image-space summaries). Bottom block: the nearest topology-aware
concurrent work (\S D.7).}
\label{tab:philosophy}
\footnotesize
\setlength{\tabcolsep}{3pt}
\begin{tabular}{lllll}
\toprule
method & repr. & philosophy & topology in objective? & topology from \\
\midrule
NeRF~\citep{mildenhall2020nerf} & density field & geometry & no & emergent \\
OccNet~\citep{mescheder2019occupancy} & occupancy & geometry & no & construction (closed iso) \\
DeepSDF~\citep{park2019deepsdf} & SDF & geometry & no & construction (closed iso) \\
NeuS~\citep{wang2021neus} & SDF & geometry & no & construction (closed iso) \\
Neuralangelo~\citep{li2023neuralangelo} & SDF & geometry & no & construction (closed iso) \\
Pixel2Mesh~\citep{wang2018pixel2mesh} & template mesh & geometry & no & frozen at template genus \\
Point2Mesh~\citep{hanocka2020point2mesh} & template mesh & geometry & no & frozen at template genus \\
Mesh R-CNN~\citep{gkioxari2019meshrcnn} & voxels $\to$ mesh & geometry & no & fixed by voxelization \\
DMTet~\citep{shen2021dmtet} & grid $+$ marching & geometry & no & extraction over the grid \\
nvdiffrec~\citep{munkberg2022nvdiffrec} & grid $+$ marching & geometry & no & extraction over the grid \\
FlexiCubes~\citep{shen2023flexicubes} & grid $+$ marching & geometry & no & extraction over the grid \\
Shape-as-Points~\citep{peng2021sap} & points $\to$ implicit & geometry & no & construction (closed iso) \\
AtlasNet~\citep{groueix2018atlasnet} & patch soup & geometry & no & unmeasured \\
DSS~\citep{yifan2019dss} & point soup & geometry & no & unmeasured \\
3DGS~\citep{kerbl2023gaussians} & Gaussian soup & geometry & no & unmeasured \\
2DGS~\citep{huang20242dgs} & surfel soup & geometry & no & unmeasured \\
SuGaR~\citep{guedon2024sugar} & Gaussians $\to$ mesh & geometry & no & post-hoc extraction \\
MILo~\citep{guedon2025milo} & Gaussians$+$mesh & geometry & no & in-loop extraction, unmeasured \\
Triangle Splatting~\citep{held2025trianglesplatting} & triangle soup & geometry & no & unmeasured \\
Radiant Tri.\ Soup~\citep{burgdorfer2025radiant} & triangle soup & geometry & no & unmeasured (proximity forces) \\
Mini-Splatting~\citep{fang2024minisplatting} & Gaussian soup & sampling & no & unmeasured (budgeted resample) \\
DiffSoup~\citep{tojo2026diffsoup} & triangle soup & geometry & no & unmeasured \\
DiffSoup $+$ prior (\S A) & triangle soup & sampling & no (steers only) & unmeasured; fails on loops \\
\midrule
Topology-GS~\citep{shen2025topologygs} & Gaussian soup & topology (image) & rendered images only & surface unmeasured \\
Gao et al.~\citep{gao2025genus,gao2026homology} & watertight mesh & geometry & no (topology a prior) & template genus / camera prior \\
STITCH~\citep{jignasu2024stitch} & implicit (points) & topology & yes ($\beta_0$, live surface) & construction $+$ constraint \\
\textbf{this work} & triangle soup & \textbf{topology} & \textbf{yes} (live-surface persistence) & measured, optimized, budgeted \\
\bottomrule
\end{tabular}
\end{table*}

\subsection{Geometry-oriented: implicit fields}
Occupancy and signed-distance learning deliver closed surfaces by
construction~\citep{mescheder2019occupancy,park2019deepsdf,peng2020convoccnet},
and a fitting stack matured around the level set: Eikonal
regularization~\citep{gropp2020igr}, surface
rendering~\citep{yariv2020idr,oechsle2021unisurf}, SDF-based volume
rendering~\citep{yariv2021volsdf,wang2021neus}, geometric
priors~\citep{yu2022monosdf}, high-resolution and accelerated
variants~\citep{li2023neuralangelo,wang2023neus2,muller2022instantngp},
meshing for deployment~\citep{yariv2023bakedsdf}, and kernel-based
point-cloud fitting~\citep{huang2023nksr}. Across all of them the
training objective is photometric or geometric; the guarantee of
consistent topology comes from the representation, its \emph{correctness}
(the right genus for the scene) is emergent, and the price is dense
field evaluation plus post-hoc extraction --- the cost regime the
soup family exists to avoid.

\subsection{Geometry-oriented: extraction, grids, and templates}
Where the surface is explicit, its topology is decided by an extraction
algorithm or frozen by a template. Marching Cubes and Dual Contouring
fix per-cell connectivity by
convention~\citep{lorensen1987marching,ju2002dual}; Poisson
reconstruction extracts an indicator iso-surface, watertight but with
emergent genus~\citep{kazhdan2006poisson,kazhdan2013screened};
differentiable iso-surfacing makes extraction trainable without making
topology a
target~\citep{remelli2020meshsdf,gao2020deftet,shen2021dmtet,munkberg2022nvdiffrec,shen2023flexicubes};
learned tessellations move the per-cell decision into a
network~\citep{chen2021nmc,chen2020bspnet}; and template or voxel-staged
pipelines fix topology before
refinement~\citep{wang2018pixel2mesh,hanocka2020point2mesh,gkioxari2019meshrcnn}.
Robust meshing gives the same guarantee-by-construction for simulation
inputs~\citep{hu2018tetwild}. In every case topology is an input or an
artifact of the algorithm --- measured nowhere, optimized never.

\subsection{Geometry-oriented: point splatting and Gaussian soups}
The connectivity-free line runs from EWA surface
splatting~\citep{zwicker2001splatting} through neural point
rendering~\citep{aliev2020npbg,ruckert2022adop,xu2022pointnerf} to the
Gaussian-splatting era~\citep{kerbl2023gaussians,yu2024mipsplatting} ---
Pulsar's abstract advertises the family's freedom from mesh ``topology
problems'' explicitly~\citep{lassner2021pulsar}. Its surface-minded
members regularize geometry ever harder --- flattened
surfels~\citep{huang20242dgs,dai2024surfels}, planar
compression~\citep{chen2024pgsr}, opacity-field
extraction~\citep{yu2024gof}, alignment-driven
meshing~\citep{guedon2024sugar} --- and MILo now extracts a mesh
differentiably at \emph{every} training
iteration~\citep{guedon2025milo}; triangle soups follow the same
pattern~\citep{held2025trianglesplatting,burgdorfer2025radiant,tojo2026diffsoup}.
Yet even with connectivity in the loop, no member of the family
measures a topological quantity of what it builds: the objectives
remain photometric fidelity plus geometric regularity.

\subsection{Sampling-oriented: steering where a fixed objective spends}
The second occupied family changes \emph{where} primitives go, not what
the objective asks of them: split-and-prune densification in
3DGS~\citep{kerbl2023gaussians}, budget-constrained Gaussian
resampling~\citep{fang2024minisplatting}, importance-driven respawn in
DiffSoup~\citep{tojo2026diffsoup} --- and the topology-inspired
allocation priors of our own study (\S A), which put the family's
implicit hypothesis to a controlled test and returned the null result
the main paper builds on. Steering allocation is a real lever for
\emph{geometric} coverage; the study shows it is structurally the wrong
lever for one-dimensional topology.

\subsection{Topology as machinery: persistence in and around learning}
The ingredients for topology-aware objectives have been maturing for two
decades: persistence computation at
scale~\citep{zomorodian2005computing,bauer2021ripser,maria2014gudhi},
practical diagram metrics~\citep{kerber2017geometry}, stability
theory~\citep{cohensteiner2007stability}, and vectorized
representations~\citep{adams2017images,bubenik2015landscapes} on the
computational side; trainable diagram
layers~\citep{hofer2017signatures,carriere2020perslay}, filtration
learning~\citep{hofer2020filtration}, and general differentiability
frameworks~\citep{leygonie2022framework,carriere2021optimizing,nigmetov2024bigsteps}
on the gradient side (surveyed
in~\citep{chazal2021introduction,carriere2026survey}). In-loop
topological losses already win in adjacent domains --- representation
learning~\citep{moor2020topoae}, tubular and topologically faithful
segmentation~\citep{shit2021cldice,stucki2023betti,hu2019topopreserving,clough2022topoloss},
microscopy volume prediction~\citep{waibel2022topological}, point-cloud
repair and shape
optimization~\citep{bruel2020toporecon,poulenard2018shapematching} --- and
differentiable non-persistence signals exist
too~\citep{roell2024dect}. What had not happened is the transfer of
this machinery into image-based surface reconstruction.

\subsection{Topology outside the loop: repair and benchmarks}
Before differentiable pipelines, topology was corrected after the fact:
spurious-handle removal~\citep{guskov2001topological,wood2004excess},
general mesh repair~\citep{attene2010repair,ju2004repair,sharf2007topologyaware}
(surveyed in~\citep{attene2013repairing}), and soup-to-manifold
conversion~\citep{huang2020manifoldplus,hu2018tetwild}; Thingi10K
quantifies how pervasive soup-like defects are in the
wild~\citep{zhou2016thingi10k} --- our tom-yum pot's raw export is a
typical specimen. The evaluation culture matches: standard pipelines
and benchmarks --- COLMAP~\citep{schonberger2016sfm,schonberger2016pixelwise},
PMVS~\citep{furukawa2010pmvs}, DTU~\citep{jensen2014dtu}, Tanks and
Temples~\citep{knapitsch2017tanks} --- score accuracy and completeness
by point distances; none carries a topological-correctness metric, the
blindness our Table~S1 quantifies on controlled probes.

\subsection{The empty cell, and its nearest neighbours}
Concurrent work has begun to press on the gap from several directions,
and the distinctions matter. Topology-GS~\citep{shen2025topologygs}
puts a persistence term \emph{inside} Gaussian-splatting training ---
but on barcodes of \emph{rendered images}, a topological-perceptual
loss; the surface's diagram is never measured, and densification adds
primitives rather than budgeting them. The high-genus inverse-rendering
line~\citep{gao2025genus,gao2026homology} reconstructs meshes whose
genus is guaranteed by a Gauss--Bonnet-matched template or guides
camera placement with homology computed once on a reference shape ---
topology as a \emph{prior}, assumed rather than measured, with
ground-truth genus required up front.
STITCH~\citep{jignasu2024stitch} does measure a differentiable
persistence quantity on a live implicit surface --- enforcing a single
connected component --- but from point-cloud input: no images, no
primitives, no budget. Topology also enters generation as conditioning
features~\citep{hu2024topodiffusion}, again outside the reconstruction
objective. Partition the three requirements --- (i) a topological
quantity of the evolving \emph{surface} measured differentiably in the
objective, (ii) image-based reconstruction through a renderer, (iii) a
fixed primitive budget --- and each neighbour holds at most one; none
holds all three. That cell --- the bottom row of
Table~\ref{tab:philosophy} --- is what the main paper fills.

\section{Per-seed values and calibration sweeps}
\label{app:perseed}

The main paper's reporting policy (\S 4) quotes Welch $\sigma$
separations only where both arms ran five seeds; every three-seed
verdict instead rests on disjoint per-seed ranges.
Table~\ref{tab:perseed} lists the underlying per-seed tail bottlenecks
for the generality wave and the three-seed analytic shapes. In every
row except rocker-arm's (the no-headroom null, main \S5.2) the largest
C1 seed is well below the smallest C0 seed; identical values across
seeds --- C2's, and the group wave's pinned C0s --- are genuine
(the unmatched-target bound, main paper \S\S5.2, 6; suppl.\ \S G), not
rounding.

\begin{table*}[t]
\centering
\caption{Per-seed tail bottlenecks (seeds 0/1/2, run order) behind the
three-seed verdicts. Spot's C1 value repeats across seeds at this
precision; the pot rows are the quicklook values behind main
Table~2.}
\label{tab:perseed}
\footnotesize
\setlength{\tabcolsep}{4pt}
\begin{tabular}{llll}
\toprule
shape (dim) & C0 & C1 loss & C2 control \\
\midrule
spot (H2)    & .0429 / .0584 / .0550 & .0237 / .0237 / .0237 & .0652 $\times$3 \\
bob (H1)     & .0429 / .0413 / .0437 & .0200 / .0207 / .0216 & .0437 $\times$3 \\
fandisk (H2) & .0540 / .0515 / .0559 & .0058 / .0055 / .0042 & .0571 $\times$3 \\
tom-yum pot (H2) & .0218 / .0223 / .0223 & .0059 / .0056 / .0056 & .0223 $\times$3 \\
rocker-arm (H1) & .0088 / .0082 / .0079 & .0090 / .0070 / .0100 & .0112 / .0112 / .0115 \\
eight (H1)   & .0249 $\times$3 & .0169 / .0175 / .0142 & .0249 $\times$3 \\
armadillo (H2) & .0511 $\times$3 & .0124 / .0125 / .0123 & .0511 $\times$3 \\
horse (H2)   & .0408 $\times$3 & .0124 / .0083 / .0118 & .0408 $\times$3 \\
cube (H2)    & .0560 / .0551 / .0636 & .0071 / .0064 / .0086 & .1346 $\times$3 \\
two spheres (H0) & .0060 / .0074 / .0062 & .0007 / .0009 / .0005 & .0009 / .0009 / .0006 \\
\bottomrule
\end{tabular}
\end{table*}

Table~\ref{tab:knobs} tabulates the two calibration-knob sweeps whose
verdicts the main paper states in \S 3.4 and \S 5.1: the $\rho$
response is flat over a $10\times$ range, and the ramp window is
immaterial --- all alternatives land within seed noise of the default.

\begin{table}[t]
\centering
\caption{Calibration-knob sweeps (sphere, tail bottleneck). \emph{Top}:
$\rho$ response, matched seed-0 sweep --- the five-seed $\rho{=}0.1$ arm
is main Table~1's C1 ($.0156{\pm}.0013$). \emph{Bottom}: ramp-window
pilot, three seeds each.}
\label{tab:knobs}
\footnotesize
\setlength{\tabcolsep}{3pt}
\begin{tabular}{lccc}
\toprule
$\rho$ (seed 0) & $0.03$ & $0.1$ & $0.3$ \\
\midrule
tail & .0181 & .0172 & .0177 \\
\midrule
ramp window & $10$--$40\%$ & $20$--$50\%$ (default) & $30$--$60\%$ \\
\midrule
tail & .0154$\pm$.0017 & .0156$\pm$.0013 & .0179$\pm$.0019 \\
\bottomrule
\end{tabular}
\end{table}

\section{The open-surface probe}
\label{app:probe}

Two open bowls are cut from the analytic sphere's ground-truth mesh by
removing a polar cap (measured openings $0.14\,R$ and $0.50\,R$), then
certified by measurement as disk topology: one shell, exactly one
boundary loop, edge-manifold, consistently oriented, Euler
characteristic $1$. The pre-registered bundle rule (the pot's) applies
unchanged: the density staircase decides the observable and density ---
and it overruled our own expectation. We predicted the wide bowl would
lose its enclosed void to the rim; instead \emph{both} bowls read
exactly one significant H2 bar at every tested density (the mouth caps
at small $\alpha$, so the interior reads as a chamber --- the pot's
mechanism on an open surface), while the \emph{designed} rim's H1 bar
never clears the floor (Table~\ref{tab:probestair}). The rule therefore
sets H2 at $M{=}2048$ for both, budget $N{=}1200$, loss restricted to
H2 (the torus precedent for sub-floor features). A designed rim is a
\emph{real} feature of the open target, distinct from the noise-born
boundary bars real scans will add --- the bar-filtering step named in
main \S7 concerns the latter.

\begin{table}[t]
\centering
\caption{Probe staircase (seed 0): top bar lifetime over the floor, per
dimension. Both bowls read exactly $(0,0,1)$ significant bars at every
density --- one H2 void, no significant H1 --- so the table gives the
margins.}
\label{tab:probestair}
\scriptsize
\setlength{\tabcolsep}{3pt}
\begin{tabular}{lccc}
\toprule
top bar / floor & $M{=}2048$ & $M{=}4096$ & $M{=}8192$ \\
\midrule
narrow: H2 / H1 & $3.79\times$ / $0.47\times$ & $5.50\times$ / $0.57\times$ & $7.82\times$ / $0.45\times$ \\
wide: H2 / H1   & $2.34\times$ / $0.42\times$ & $3.50\times$ / $0.51\times$ & $4.94\times$ / $0.46\times$ \\
\bottomrule
\end{tabular}
\end{table}

The C-matrix (C0/C1/C2, three seeds, protocol otherwise unchanged)
passes the verdict rule on both bowls at better-than-baseline Chamfer
(Table~\ref{tab:probe}): the loss cuts the void's error $4.2\times$
(narrow) and $4.1\times$ (wide) with disjoint per-seed ranges, inside
the study's $1.5$--$10.4\times$ span. The control replicates its
closed-surface failure modes: on the narrow bowl it pins the void at
its repulsion equilibrium ($2.1\times$ worse than baseline, identical
across seeds) at $3.5\times$ worse Chamfer; on the wide bowl it
\emph{erases} the void in two of three seeds ($\beta_2\,1{\to}0$) at
$3.1\times$ worse Chamfer. Every C0/C1 seed reads the correct count.

\begin{table}[t]
\centering
\caption{Open-surface probe C-matrix: tail bottleneck (mean$\pm$sd, 3
seeds; per-seed values in run order beneath); $\times$ = reduction vs
C0. Both shapes \textbf{pass} the verdict rule; no Welch $\sigma$ at
$n{=}3$ (main \S4).}
\label{tab:probe}
\scriptsize
\setlength{\tabcolsep}{2pt}
\resizebox{\linewidth}{!}{%
\begin{tabular}{llll}
\toprule
shape & C0 & C1 loss & C2 control \\
\midrule
bowl$_{\mathrm{narrow}}$ & .0588$\pm$.0035 & \textbf{.0138$\pm$.0013} (4.2$\times$) & .1249 (pinned) \\
\;\footnotesize per-seed & \footnotesize .0619/.0550/.0593 & \footnotesize .0137/.0126/.0151 & \footnotesize .1249 $\times$3 \\
bowl$_{\mathrm{wide}}$ & .0560$\pm$.0098 & \textbf{.0137$\pm$.0003} (4.1$\times$) & .0770 ($\beta_2\,1{\to}0$, 2/3) \\
\;\footnotesize per-seed & \footnotesize .0448/.0602/.0630 & \footnotesize .0141/.0136/.0135 & \footnotesize .0770 $\times$3 \\
\bottomrule
\end{tabular}}
\end{table}

\section{The group wave: expanding the external set to two groups}
\label{app:groupwave}

The generality wave of main \S5.2 grew from four external meshes to
eight, organized as the study's two observable classes --- the
\emph{loop group} (H1: bob, rocker-arm, eight) and the \emph{void
group} (H2: spot, fandisk, the tom-yum pot, armadillo, horse) --- under
a pre-registered plan (\texttt{GROUP\_WAVE\_PLAN.md}, shipped with the
release bundle) that fixed the candidate pool, the selection rule, the
run protocol, and the group statistic before any training run. Per-seed
tails for the new shapes join suppl.\ Table~S7; this section records
where the shapes came from and how they were selected, and
Fig.~\ref{fig:gwmatrix} renders all four under the front-page
convention.

\begin{figure*}[t]
\centering
\includegraphics[width=0.72\textwidth]{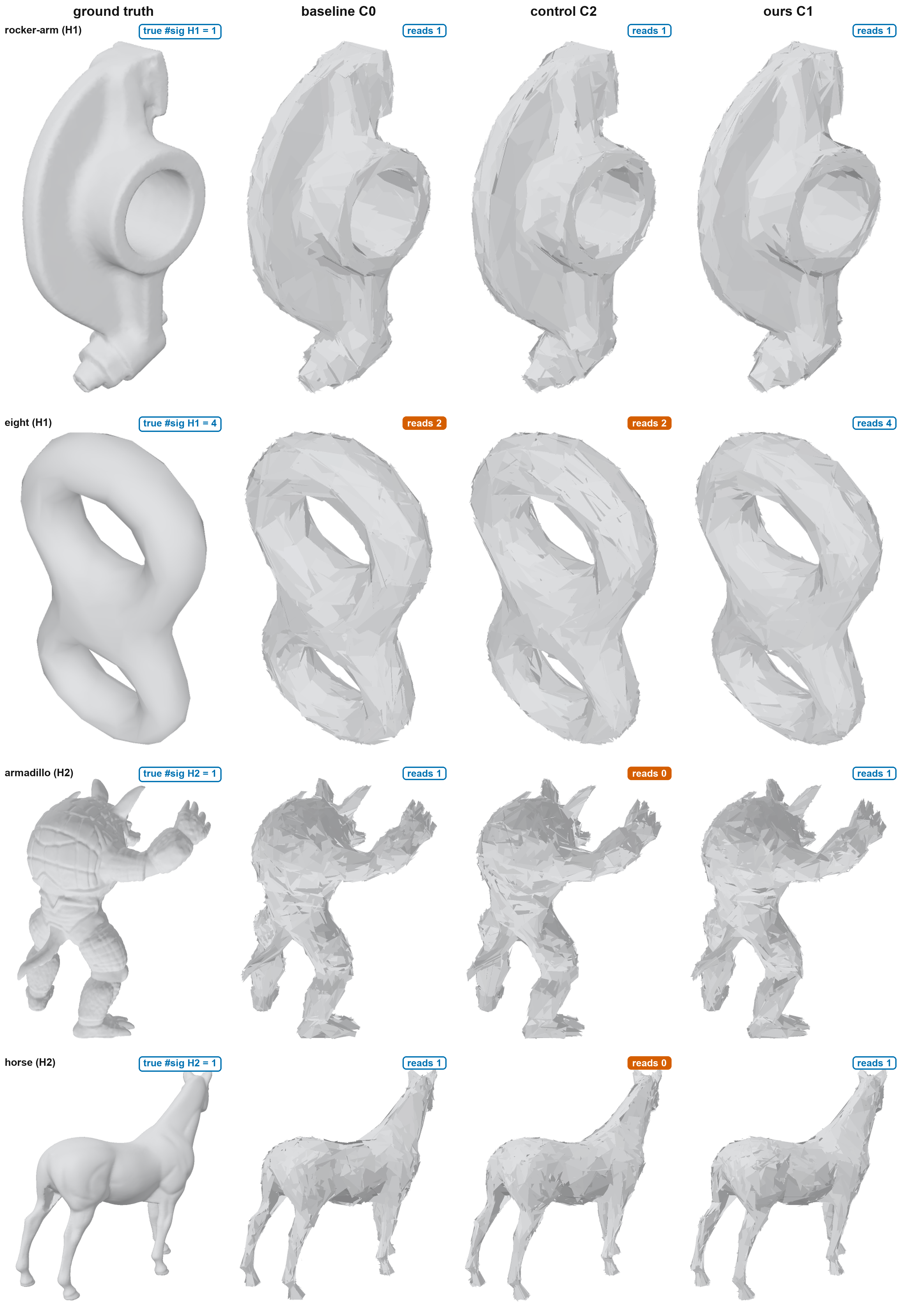}
\caption{The group wave, rendered (ground truth / C0 / C2 / C1; seed-0
final soups; material, per-row camera and badge convention of main
Fig.~1). \emph{rocker-arm}: the no-headroom null --- every arm reads
the one measurable loop correctly (all badges blue; the control's
damage here is value-only). \emph{eight}: baseline and control read
only two of the four certified loops while the loss restores all four
(per-seed C1 counts $4/3/4$; seed 0 shown) --- and the missing pair is
invisible in the render, both holes being open in every cell: the
deficit is a sub-floor diagram reading, exactly the blindness the
paper measures. \emph{armadillo}/\emph{horse}: correct-count baselines
whose void \emph{value} pins at the unmatched bound (main Table~2);
the control erases both voids (reads 0) with no visible geometric
collapse.}
\label{fig:gwmatrix}
\end{figure*}

\paragraph{Pool, certificates, honest exclusions (Table~\ref{tab:gwdisp}).}
Candidates were drawn from a standard test-model collection
(\texttt{alecjacobson/common-3d-test-models}, pinned commit) plus the
CGAL repository's genus-2 \emph{eight}; the release bundle ships fetch
scripts with pinned commits and the certificates, not the meshes ---
redistribution terms remain the original repositories'.
Every candidate passed or failed the pot's mesh certificate (main \S4:
exact simplicial homology cross-checked against per-component Euler
characteristics, watertightness, orientation) as-is --- failures were
excluded, never repaired, keeping the wave's class "watertight
benchmark" (the pot remains the study's one ingest-and-solidify story).
Five of eleven failed: open boundaries or non-manifold junctions in
beast, cow, max-planck, ogre, and the Utah teapot.

\paragraph{Staircase selection, before training (Table~\ref{tab:gwstair}).}
The density staircase (the pot's pre-registered rule, unchanged) fixed
each certified shape's observable class and working density:
\emph{rocker-arm} (genus 1) enters the loop group on the torus
precedent --- one of its two loops clears the $M{=}2048$ floor
($0.0929$ vs $0.0613$), the second ($0.0204$) reaches only
$1.15\times$ the floor even at $M{=}20{,}000$, and the void is
sub-floor at 2048, so the loss is H1-restricted;
\emph{eight} (genus 2) enters on the pot's precedent --- the full
exact $(1,4,1)$ is significant from $M{=}4096$, and the bundle is
built at $M{=}8192$ where the fourth loop clears $1.69\times$ and the
void $2.6\times$; \emph{armadillo} and \emph{horse} (genus 0) read the
exact $(1,0,1)$ at every tested density (voids $1.68\times$/$1.42\times$
over the 2048 floor).
Two certified candidates were set aside at this stage, before any
training: \emph{cheburashka} qualifies at 2048 but grows a second
significant capping H2 bar from $M{=}8192$ ($0.0300$ vs floor $0.0271$;
$0.0312$ vs $0.0171$ at 20k) --- a pot-flue-type concavity reading that
would make the evaluation-density phantom check ambiguous --- and a
denser retessellation of eight (license unstated at its source) whose
staircase independently corroborates the $(1,4,1)$ observable at
$M{=}8192$.

\paragraph{Outcome identities behind the pins.}
The three seed-exact baseline values of main Table~2 are unmatched-target
bounds, not coincidences: eight's C0/C2 pin $0.0249 \approx$ half the
third loop's ground-truth lifetime at the evaluation density
($0.0498/2$), armadillo's $0.0511 = 0.1022/2$, horse's
$0.0408 \approx 0.0817/2$ --- the bottleneck cost of leaving that
target bar unmatched (main \S6's pin mechanism, previously seen in the
pot's control arm).

\begin{table}[t]
\centering
\caption{Group-wave staircases (seed 0): significant Betti counts vs
sampling density for the four selected shapes; bold = the pre-registered
working density. rocker-arm's second loop and 2048-density void never
clear the floor (loss H1-restricted, torus precedent); eight's full
$(1,4,1)$ holds from 4096 (bundle at 8192, pot precedent).}
\label{tab:gwstair}
\scriptsize
\setlength{\tabcolsep}{4pt}
\begin{tabular}{lcccc}
\toprule
sig $(\beta_0,\beta_1,\beta_2)$ & $M{=}2048$ & $4096$ & $8192$ & $20000$ \\
\midrule
rocker-arm & \textbf{(1,1,0)} & (1,1,1) & (1,1,2) & (1,2,2) \\
eight      & (1,2,1) & (1,4,1) & \textbf{(1,4,1)} & (1,4,1) \\
armadillo  & \textbf{(1,0,1)} & (1,0,1) & (1,0,1) & (1,0,1) \\
horse      & \textbf{(1,0,1)} & (1,0,1) & (1,0,1) & (1,0,1) \\
\bottomrule
\end{tabular}
\end{table}

\begin{table}[t]
\centering
\caption{Candidate disposition (selection fixed before training;
plan file in the release bundle). Certificate failures are excluded,
never repaired; the two staircase-stage set-asides are motivated in
the text.}
\label{tab:gwdisp}
\scriptsize
\setlength{\tabcolsep}{3pt}
\resizebox{\columnwidth}{!}{%
\begin{tabular}{llll}
\toprule
candidate & origin & certificate & decision \\
\midrule
rocker-arm & INRIA / AIM@SHAPE & pass (g1) & loop group \\
eight & CGAL data (CC0-class) & pass (g2) & loop group \\
armadillo & Stanford scan rep. & pass (g0) & void group \\
horse & CyberWare / GaTech & pass (g0) & void group \\
cheburashka & I.~Baran (Pinocchio) & pass (g0) & set aside \\
eight (dense) & Directional (lib.) & pass (g2) & set aside \\
beast & Autodesk & \textbf{fail}: non-manifold & --- \\
cow & Viewpoint/Sun & \textbf{fail}: odd $\chi$ & --- \\
max-planck & MPI & \textbf{fail}: 161 open edges & --- \\
ogre & K.~Ritchie & \textbf{fail}: 200 open edges & --- \\
teapot & M.~Newell & \textbf{fail}: 160 open edges & --- \\
\bottomrule
\end{tabular}%
}
\end{table}

\section{Robustness of the reporting choices}
\label{app:robustness}

\begin{figure*}[t]
\centering
\includegraphics[width=\textwidth]{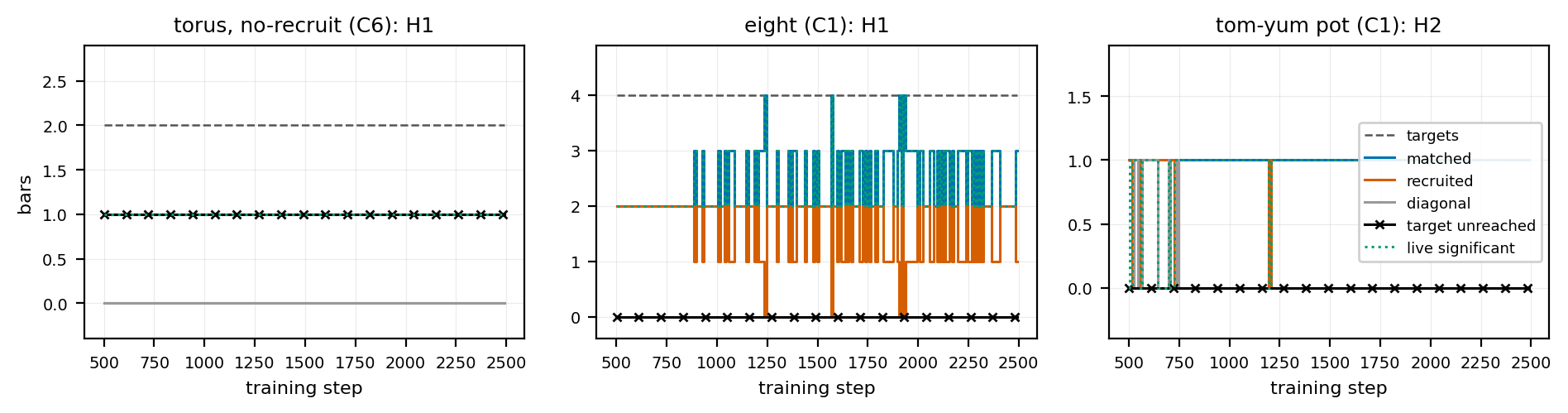}
\caption{Per-refresh plan composition from the recorded refresh logs
(seed 0; ${\sim}200$ refreshes each; no new runs). \emph{Left}: torus
without recruitment (C6) --- one loop matched, one target unreached at
every refresh: the chronic gap behind the C6 collapse (main \S5.1).
\emph{Middle}: eight (C1) --- recruitment claims the two sub-threshold
loops early and hands them to the matched set as training makes them
significant. \emph{Right}: the tom-yum pot (C1) --- a stable one-bar
matched plan after the first ${\sim}250$ active steps. Dashed gray =
target count; in the right panel it coincides with the matched line.}
\label{fig:pairstab}
\end{figure*}

\paragraph{The floor multiplier is not load-bearing.}
The significance floor $k\,r_{\mathrm{med}}(M)$ uses $k{=}6$ throughout.
Recomputing every documented bundle decision from the recorded
ground-truth clouds (protocol densities, bundle seed) gives each
feature's flip point --- the $k$ at which an included feature would
leave the bundle or an excluded one would enter --- and \emph{every
decision in the study is unchanged for
$k \in (5.12,\, 6.77)$} (Table~\ref{tab:floorwin}). The two binding
features are exactly the ones the paper already flags as marginal: the
double torus's sub-floor tube loops ($0.85\times$ the $k{=}6$ floor)
bound it below, bob's fat-tube void ($1.13\times$) bounds it above.
The verdict metric itself (bottleneck to the target diagram) contains
no threshold; $k$ only gates which bars enter bundles and counts.

\begin{table}[t]
\centering
\caption{Floor-multiplier windows per shape (recorded GT clouds at the
protocol bundle density): the tightest included margin and, where the
protocol excludes sub-floor features, the largest excluded margin ---
each decision holds while $k$ stays below/above $6\times$ that margin.
Global window: $k \in (5.12, 6.77)$.}
\label{tab:floorwin}
\scriptsize
\setlength{\tabcolsep}{2.5pt}
\begin{tabular}{lcc}
\toprule
shape & tightest IN ($\times$ floor) & largest OUT ($\times$ floor) \\
\midrule
sphere / cube    & 3.86 / 2.71 & --- \\
torus            & 1.33 (loop 2) & 0.83 (void) \\
two spheres      & 2.61 & --- \\
double torus     & 2.23 & 0.85 (tube loops) \\
spot / fandisk   & 1.87 / 1.50 & --- \\
bob              & 1.13 (void) & --- \\
tom-yum pot      & 1.22 & 0.46 (loops) \\
rocker-arm       & 1.52 & 0.74 (void) \\
eight            & 1.69 (loop 4) & --- \\
armadillo / horse & 1.68 / 1.42 & --- \\
bowls (narrow/wide) & 3.79 / 2.34 & 0.47 / 0.42 (rims) \\
\bottomrule
\end{tabular}
\end{table}

\paragraph{The pair-frozen plan is stable across refreshes
(Fig.~\ref{fig:pairstab}).}
Every run since the refresh-log schema records the plan's composition
(matched / diagonal / recruited / target-unreached, per dimension) at
each of its ${\sim}200$ refreshes. Three regimes cover the study:
the torus without recruitment (C6) shows the chronic gap the ablation
collapses on --- one matched loop and one unreached target at
\emph{every} refresh; on the genus-2 eight, recruitment claims the two
sub-threshold loops from the first active refresh and \emph{hands them
over} to the matched set as training makes them significant (matched
$2{\to}3$--$4$ late, recruited $2{\to}1$--$0$, unreached $0$
throughout); the tom-yum pot settles to a stable one-bar matched plan
within the first ${\sim}250$ active steps. No plan oscillation beyond
single-bar hand-overs is observed, and the per-refresh Gabriel-failure
count is zero in all three runs, as everywhere (main \S3.2).

\paragraph{Cost scaling, measured.}
The persistence step --- the refresh's dominant cost --- measured on
the torus cloud: $174$\,ms at $M{=}2048$ (reproducing \S3.2's recorded
refresh figure), $388$\,ms at $4096$, $740$\,ms at $8192$, $1.92$\,s at
$20{,}000$ --- an $11.0\times$ cost for $9.8\times$ the points,
consistent with the near-linear $O(M\log M)$ scaling.
GPU memory: a single-seed probe (torus, $2{,}500$ steps, whole-GPU used
memory sampled at $400$\,ms) peaks at $3{,}649$\,MiB for the baseline
and $4{,}035$\,MiB with the loss ($+386$\,MiB); the persistence
computation itself runs on the CPU (GUDHI), off the GPU budget.

\end{document}